\documentclass[11pt,reqno,oneside]{amsart}

\usepackage[a4paper, total={6in, 9in}]{geometry}
\usepackage{graphicx}
\usepackage{amssymb}
\usepackage{epstopdf}
\usepackage{amsmath,amsthm,mathrsfs}
\usepackage[usenames]{color}
\usepackage{xcolor}
\usepackage{bm}
\usepackage{subcaption}
\usepackage{array}
\usepackage{makecell}
\usepackage{longtable}
\usepackage{booktabs}
\usepackage{paralist}
\usepackage{verbatim}
\usepackage{tabularx}
\usepackage{enumitem}
\usepackage{hyperref}

\hypersetup{
  colorlinks=true,
  linkcolor=blue,
  citecolor=blue,
  urlcolor=blue
}

\newtheorem{thm}{Theorem}[section]

\newtheorem{lem}{Lemma}[section]
\newtheorem{prop}{Proposition}[section]
\theoremstyle{definition}
\newtheorem{defn}{Definition}[section]

\theoremstyle{remark}
\newtheorem{rem}{Remark}[section]
\newtheorem{remark}{Remark}
\numberwithin{equation}{section}

\newcommand{\beq}{\begin{equation}}
\newcommand{\eeq}{\end{equation}}

\title[Effective medium theory for embedded sound-soft obstacles in media]
{Effective medium theory for embedded sound-soft obstacles in an anisotropic inhomogeneous medium with applications}

\author[Huaian Diao]{Huaian Diao}
\address{School of Mathematics and Key Laboratory of Symbolic Computation and Knowledge Engineering of Ministry of Education, Jilin University, Changchun 130012, China}
\email{diao@jlu.edu.cn}

\author[Qingle Meng]{Qingle Meng*}
\thanks{*Corresponding author}
\address{Department of Mathematics, City University of Hong Kong, Kowloon Tong, Hong Kong SAR, China}
\email{mengql2021@foxmail.com, qinmeng@cityu.edu.hk}

\author[Zhiying Sun]{Zhiying Sun}
\address{School of Mathematics, Jilin University, Changchun, Jilin, China}
\email{sunzy23@mails.jlu.edu.cn}

\begin{document}

\maketitle

\begin{abstract}
This paper investigates the problem of time-harmonic acoustic scattering in an inhomogeneous medium with a complex topological structure. Specifically, the medium is anisotropic and contains several disjoint sound-soft obstacles. This model commonly arises in the inverse scattering problem of simultaneously recovering the embedded obstacles and the surrounding medium. We propose a novel theoretical framework that demonstrates how embedded obstacles can be effectively approximated by an isotropic and lossy medium with specified physical parameters, wherein the total wave field exhibits decay properties related to these specified material parameters at the boundaries of the obstacles. This mathematical characterization of the wave in an effective medium model can be used to locate the underlying obstacles. Furthermore, we establish rigorous estimates to validate this approximation and provide a concrete example illustrating our theoretical results. Our proposed effective medium theory offers substantial applications within the context of the aforementioned inverse problem.

\medskip
		

\medskip

\end{abstract}

\section{Introduction}
	\subsection{Mathematical setup}
%
Let $D_1, D_2, \ldots, D_n$ ($n \in \mathbb{N}$) represent $n$ sound-soft obstacles that are pairwise disjoint. These impenetrable obstacles are situated within a bounded Lipschitz domain \( \Omega \subset \mathbb{R}^N\) $(N=2,3)$. The union $G := \bigcup_{i=1}^n D_i$ is such that the complement $\Omega \setminus \overline{G}$ is connected; see Fig.~\ref{Fig1} for a schematic illustration.
 We assume that the medium in $\Omega \setminus \overline{G}$ is inhomogeneous, anisotropic, and lossy, while the medium outside $\Omega$ is homogeneous and free of damping. Additionally, in the domain $\Omega \setminus \overline{G}$, we define a matrix-valued function $A: \Omega \setminus \overline{G} \to \mathbb{C}^{N \times N}$, denoted as $A = (a_{jk})_{j,k=1}^N$, where the entries $a_{jk} \in C^1(\Omega \setminus \overline{G})$ are continuously differentiable functions. The real part of $A$, denoted by $\Re(A)$, is a matrix-valued function whose entries are the real parts of $a_{jk}$, i.e., $\Re(a_{jk})$. Similarly, the imaginary part of $A$, denoted by $\Im(A)$, is a matrix-valued function with entries $\Im(a_{jk})$. Furthermore, we assume that, for each $x \in \Omega \setminus \overline{G}$, the matrices $\Re(A)$ and $\Im(A)$ are symmetric and satisfy the properties outlined in \cite{ref4} for all $\xi\in\mathbb{C}^N$,
	\begin{equation}
		\bar{\xi}\cdot\Im (A) \xi \leq 0\quad\mbox{and}\quad \bar{\xi}\cdot\Re (A) \xi\geq\beta|\xi|^{2},\label{equ:10.1}
	\end{equation} 
	where $\beta$ is a positive constant. Due to the symmetry of $A$, it follows that
	\begin{equation}\label{A2}
    \Re\left(\bar{\xi}\cdot A \xi\right)=\bar{\xi}\cdot \Re(A) \xi\quad\mbox{and}\quad  \Im\left(\bar{\xi}\cdot A \xi\right)=\bar{\xi}\cdot \Im(A) \xi.
    \end{equation}
  
   Denote  
    $$
    c(x)=\begin{cases}
    \widetilde {c}(x) & \mbox{in } \Omega \setminus \overline G, \\
    c_0&  \mbox{in } \mathbb R^N \setminus \overline \Omega
    \end{cases}\quad \mbox{and}\quad \sigma(x) =\begin{cases}
    \widetilde {\sigma}(x) & \mbox{in } \Omega \setminus \overline G, \\
    0&  \mbox{in } \mathbb R^N \setminus \overline \Omega,
    \end{cases}
    $$
where $c(x)$ is the wave speed in the acoustic medium $\Omega\setminus \overline G $ and homogenous background respectively,  and the nonnegative function $\sigma(x)$ is the damping coefficient of the medium $\Omega\setminus \overline G $. Here $\widetilde {c}(x)\in L^\infty(\Omega \setminus \overline G )$, $\widetilde {\sigma}(x)\in L^\infty(\Omega \setminus\overline G )$, and $c_0\in \mathbb R_+$ is a constant. 

Our study is motivated by the investigation of acoustic wave propagation in an anisotropic medium $\Omega \setminus \overline G$ embedded within an isotropic background medium $\mathbb R^N \setminus \Omega$, where $G$ is assume to be sound-soft obstacles. This phenomenon is governed by the wave equation (see, e.g., \cite{ref2, LAX1967}):
\begin{equation}\label{eq:waveequation}
\frac{1}{c^2(x)}\frac{\partial ^2\widetilde{u}(x,t)}{\partial t ^2}+\sigma(x)\frac{\partial \widetilde{u}}{\partial t}-\nabla\cdot \left(\widetilde A(x)\nabla \widetilde{u}\right)=\widetilde{f}(x,t),\quad x\in \mathbb{R}^N\backslash\overline{G},\quad t>0,
\end{equation}
where \(\widetilde{u}(x,t)\) is the wave field, and \(\widetilde{f}(x,t)\) is the source supported in \(\Omega\) with respect to the spatial variable \(x\). Here
$$
\widetilde A(x)=\begin{cases}
	A(x) & \mbox{in } \Omega \setminus \overline G,\\
	I & \mbox{in } \mathbb R^N \setminus\overline \Omega.
\end{cases}
$$

For the scattering of time-harmonic waves by obstacles within the medium described above, the time-harmonic wave field and source term are assumed to have the form
 \begin{equation}\label{eq:time-harmonic}
 \widetilde{u}(x,t)=\Re\{u(x)e^{-i\omega t}\}\quad\mbox{and}\quad \widetilde{f}(x,t)=\Re\{(-f(x))e^{-i\omega t}\},
\end{equation}
where $\omega\in \mathbb{R}_+$ denotes the angular frequency. By substituting these expressions into \eqref{eq:waveequation} and simplifying, the physical constraints on the density function $A(x)$, wave speed $c(x)$, and damping coefficient $\sigma(x)$ naturally lead to the formulation of $\gamma(x)$ and $q(x)$ as follows
	\begin{equation}\label{eq:qq1}
\gamma(x)=\begin{cases}A(x)&x\in\Omega\backslash\overline{G},\\[2ex]1&x\in\mathbb{R}^N\backslash\overline{\Omega}\end{cases}\quad \mbox{and}\quad 
q(x)=\begin{cases}\eta (x)+i\tau (x) &x\in\Omega\backslash\overline{G},\\[2ex]1&x\in\mathbb{R}^N\backslash\overline{\Omega},\end{cases} 
    \end{equation}
    where 
    $$
    \mbox{$\eta(x) = \frac{c_0^2}{\widetilde c^2(x) }  $ and $\tau(x)=\frac{c_0\widetilde \sigma (x) }{\omega}$.}
    $$
    It is clear that $q(x) \in L^{\infty}(\Omega \setminus \overline{G})$ and that there exists a positive constant $\vartheta_0$ satisfying the following conditions for any $x \in \Omega\backslash\overline{G}$
	\begin{equation*}
		\Re q(x) \geq \vartheta_0 \quad \mbox{and} \quad \Im q(x) \geq 0.
	\end{equation*}

%

\begin{figure}[ht]
\centering
\includegraphics[width=0.25\textwidth]{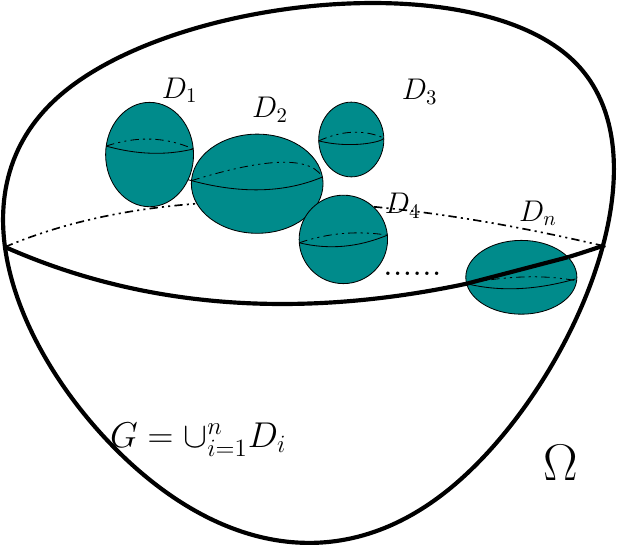}
    \caption{Schematic illustration of the complex scatterer involving the embedded obstacles and the surrounding anisotropic medium in 3D}\label{Fig1}
\end{figure}


 In this paper, we consider the time-harmonic incident field defined as $u^i=e^{ikx\cdot d}$ for $x\in\mathbb{R}^{N}$, with $d\in \mathbb{S}^{N-1}$ denoting the direction of propagation and $k$ being the wave number. Due to the presence of the complex scatterer, denoted by $G\oplus(\Omega\backslash
	\overline{G}; \gamma, q)$, the acoustic scattered field  $u^s\in H_{\mathrm{loc}}^{1}(\mathbb{R}^{N}\backslash\overline{\Omega})$, 
	characterizes the perturbation of the incident field's propagation. The total wave field is expressed as
	$u:=u^i+u^s$.  With all these preparations, substituting \eqref{eq:time-harmonic} into \eqref{eq:waveequation}, using the notation given by \eqref{eq:qq1}, we have the time-harmonic scattering problem associated with the complex scatterer $G\oplus(\Omega\backslash
	\overline{G}; \gamma, q)$,
	\begin{equation}\label{equ:1.2}
		\begin{cases}
			\nabla\cdot(A\nabla u)+k^2q(x)u=0&\text{in}   ~~\Omega\setminus\overline{G}, \\ \Delta u^s+k^2u^s=f&\text{in} ~~\mathbb{R}^N\backslash\overline{\Omega},\\ u=u^{i}+u^s&\text{in} ~~\mathbb{R}^N\backslash\overline{\Omega},\\ u=0&\text{on} ~~\partial G,\\ u^-=u^s+u^{i}\quad\frac{\partial u}{\partial\nu_{A}}=\frac{\partial u^s}{\partial\nu}+\frac{\partial u^{i}}{\partial\nu}&\text{on}~~ \partial\Omega,\\[1.5ex]\lim\limits_{|x| \to \infty}|x|^{(N-1)/2}\big(\frac{\partial u^s}{\partial|x|}-iku^s\big)=0,
		\end{cases}
	\end{equation} 
	where $f(x)$ is a source that is compactly supported outside $\Omega$ (i.e., $\mathrm{supp}f(x)\subset B_{r_0}\backslash\overline\Omega $ for some $r_0>0$). Here, $B_{r_0}$ denotes a ball of radius $r_0$ centered at the origin in $\mathbb{R}^N$, and $\nu$ represents the exterior unit normal to $\partial \Omega$. The notation $\frac{\partial u}{\partial\nu_{A}}:=\lim\limits_{h\to\infty}\nu(x)\cdot A(x)\nabla u\left(x-h\nu(x)\right)$ defines the normal derivative of
	$u$ with respect to the matrix $A$. Clearly, in  (\ref{equ:1.2}), 
	the total field $u$ vanishes on the boundary  $\partial G=\bigcup_{i=1}^{n} \partial D_i.$ The last equation in (\ref{equ:1.2}) is known as the Sommerfeld radiation condition. Furthermore, the solution admits the following asymptotic behavior:
	\begin{equation*}
		u^s(x)=\frac{e^{ik|x|}}{|x|^{(N-1)/2}}\left \{{u^{\infty} (\hat{x},d,k)+\mathcal{O}\left(\frac{1}{|x|}\right)}\right \}\quad \mbox{as}\quad |x|\to\infty,
	\end{equation*}
	uniformly in all directions $\hat{x}:=\frac{x}{|x|}\in\mathbb{S}^{N-1}$, where $d$ is defined previously. $u^{\infty}\left(\hat{x},d,k\right)$  is known as the far-field pattern, also referred to as the scattering amplitude \cite{ref2}.  
	
	\begin{rem}
	The mathematical model presented in (\ref{equ:1.2}) has been studied in the context of cloaking, particularly through the design of invisibility layers (see \cite{KOVW2020, Liu2009, LiuSun2013, Nguyen2010}), where the obstacle $G$ consists of a single component. However, our work differs fundamentally from these studies in two key aspects. First, in this paper we focus on recovering obstacles $G$ embedded in the anisotropic medium $\Omega \setminus \overline{G}$, whereas the cited works aim to cloak a single inclusion. Second, the cloaked inclusion in those studies is typically assumed to be small in size with a simply connected topology, and its cloaking schemes rely heavily on the geometry of the inclusion. In contrast, the sound-soft obstacles $G$ in our study are significantly more general, comprising disjoint obstacles $D_i$ with arbitrary Lipschitz boundaries. Notably, these obstacles $D_i$ can vary widely in scale, ranging from large to standard or small sizes relative to the wavelength of the incident wave.
	\end{rem}

	An inverse scattering problem associated with (\ref{equ:1.2}) involves the simultaneous recovery of buried obstacles, represented by $G\oplus(\Omega\backslash\overline{G}; \gamma, q)$, utilizing the corresponding far-field pattern $u^{\infty} (\hat{x},d,k)$. This problem is formulated as follows:
	\begin{equation} \label{equ:1.4}
		u^\infty(\hat{x};u^{i},f,G\oplus(\Omega\backslash\overline{G};\gamma,q))\to {G}\oplus(\Omega\backslash\overline{G};\gamma,q).
	\end{equation} 
	This paper investigates the inverse problem from a novel perspective that is applicable to a very generic scenario. Specifically, we aim to approximate  (\ref{equ:1.4}) by applying the effective medium theory and carefully selecting physical parameters to treat buried obstacles as a medium. To achieve this, we introduce the following definition.
	\begin{defn}\label{def:1.1}
		Let $u^\infty(\hat{x};u^{i},f,G\oplus(\Omega\backslash\overline{G};\gamma,q))$ be the far-field pattern corresponding to (\ref{equ:1.2}) with sound-soft obstacles. If there exists a medium $(\Omega;\gamma_{\varepsilon(x)},q_{\varepsilon(x)})$ such that $(\gamma_{\varepsilon(x)},q_{\varepsilon(x)})\big|_{\Omega\setminus\overline{G}}=(\gamma,q)\big|_{\Omega\setminus\overline{G}}$ and a small parameter $\varepsilon\in\mathbb{R}_{+}$, with $\varepsilon\ll1$,  that satisfies the following inequality
		\begin{align*}
			&\left\|u^{\infty}(\hat{x};u^{i},f,(\Omega;\gamma_{\varepsilon(x)},q_{\varepsilon(x)})-u^{\infty}(\hat{x};u^{i},f,G\oplus(\Omega\setminus\overline{G};\gamma,q))\right\|_{C(\mathbb{S}^{N-1})} \\ &\quad\leq C\varepsilon\big(\|u^{i}\|_{H^{1}(B_{r_{0}})}+\|f\|_{L^{2}(B_{r_{0}})}\big),
		\end{align*}where $B_{r_0}$ denotes any given ball centered at the origin that contains $\Omega$, and let $C$ be a positive constant depending on the a priori parameters, then $(\Omega;\gamma_{\varepsilon(x)},q_{\varepsilon(x)})$ is referred to as an effective $\varepsilon$-realization of $G\oplus(\Omega\setminus\overline{G})$. Here, $u^{\infty}(\hat{x};u^{i},f,(\Omega;\gamma_{\varepsilon(x)},q_{\varepsilon(x)})$ represents the far-field pattern associated with the scattering problem (\ref{equ:1.2}), where the sound-soft obstacles are replaced by certain penetrable mediums.
	\end{defn}

In what follows, we construct a medium scattering system in which the medium in $G$, with specially chosen material parameters $\gamma_{\varepsilon(x)}$ and $q_{\varepsilon(x)}$, serves as an effective realization of the embedded obstacles in the sense of Definition \ref{def:1.1}. This medium scattering system is rigorously characterized as follows: we seek $u_\varepsilon \in H_{loc}^1(\mathbb{R}^N)$ satisfying
	\begin{equation} \label{equ:1.5}  
		\begin{cases}\nabla\cdot(\gamma_{\varepsilon(x)}\nabla u_{\varepsilon})+k^{2}q_{\varepsilon(x)}u_{\varepsilon}=0&\text{in} ~\Omega,\vspace*{1mm}\\\Delta u_{\varepsilon}^{s}+k^{2}u_{\varepsilon}^{s}=f&\text{in}~ \mathbb{R}^{N}\setminus\overline{\Omega},\\u_{\varepsilon}=u^{i}+u_{\varepsilon}^{s}&\text{in} ~\mathbb{R}^{N}\setminus\overline{\Omega},\\ u_{\varepsilon}^{-}=u_{\varepsilon}^{+},\quad\varepsilon^{-1} \frac{\partial u_{\varepsilon}^{-}}{\partial\nu}=\frac{\partial u_{\varepsilon}^{+}}{\partial\nu_{A}}&\text{on}~\partial G,\vspace*{2mm}\\u_{\varepsilon}^{-}=u_{\varepsilon}^{s}+u^{i},\quad \frac{\partial u_{\varepsilon}^{-}}{\partial\nu_{A}}=\frac{\partial u_{\varepsilon}^{s}}{\partial\nu}+\frac{\partial u^{i}}{\partial\nu}&\text{on}~\partial\Omega,\vspace*{2mm}\\ \lim\limits_{ |x|\to\infty} ~|x|^{(N-1)/2}\left\{\frac{\partial u_{\varepsilon}^{s}}{\partial|x|}-iku_{\varepsilon}^{s}\right\}=0,\end{cases}
	\end{equation}
	where $\gamma_{\varepsilon(x)}$ and $q_{\varepsilon(x)}$ are defined as
	\begin{equation}\label{eq:gamma}
		\gamma_{\varepsilon(x)}=\begin{cases}\varepsilon^{-1}&x\in G,\\[2mm]A(x)&x\in\Omega\backslash\overline{G}\end{cases}\quad \mbox{and}\quad q_{\varepsilon(x)}=\begin{cases}\eta_0+i\varepsilon^{-1}\tau_0&x\in G,\\[2mm]
			\eta(x)+i\tau (x)&x\in\Omega\backslash\overline{G}.\end{cases} 
	\end{equation}
In this context, the notation $u_\varepsilon^-$ and $u_\varepsilon^+$ denote the limits of $u_\varepsilon$ approaching from the interior and the exterior, respectively, at the boundary $\partial G$ or $\partial \Omega$. 
	\begin{remark}
The parameters
 $ \eta_0$ and $\tau_0$ in \eqref{eq:gamma} are consistent with those given in \eqref{eq:qq1}, and $\varepsilon$ is a sufficiently small constant. When considering the obstacle $G$ as a medium, the notations $\gamma_{\varepsilon(x)}$ and $q_{\varepsilon(x)}$  represent the corresponding physical parameters within $\Omega$, where $A(x)$ satisfies \eqref{equ:10.1} and \eqref{A2}. Furthermore, the imaginary part of $q_{\varepsilon(x)}$ within $G$  differs from that of the surrounding medium by $\varepsilon^{-1}$, indicating that the material in $G$ exhibits very large damping. Notably, the behaviors of $\gamma_{\varepsilon(x)}$ and $q_{\varepsilon(x)}$ within $G$ differ significantly from those in $\Omega\backslash\overline{G}$.
\par For the inverse scattering problem  (\ref{equ:1.4}), we shall demonstrate that as $\varepsilon \to 0^+$, the far-field pattern of the solution $u_\varepsilon$ to (\ref{equ:1.5}) effectively approximates the one corresponding to the solution $u$ to (\ref{equ:1.2}), and the solution $u_\varepsilon$ to (\ref{equ:1.5}) exhibits decay on $\partial G$. These results are stated in the following theorems.
	\end{remark} 
	\begin{thm}\label{thm:1.1}
		Assume that $u\in H_{loc}^1(\mathbb{R}^N\backslash\overline{G})$ and $u_\varepsilon\in H_{loc}^1(\mathbb{R}^N)$ are solutions to (\ref{equ:1.2}) and (\ref{equ:1.5}), respectively. Let $B_r$ and $B_{r_0}$ be any central balls satisfying $\Omega\subset B_{r_0} \subset B_r$, then the medium $(G;\gamma_{\varepsilon(x)}|_G,q_{\varepsilon(x)}|_G)$ is  an $\varepsilon^{1/2}$-realization of obstacles $G$ in the sense of Definition \ref{def:1.1}. Namely, the following inequality holds for sufficiently small $\varepsilon<\varepsilon_{0}$, where $\varepsilon_{0}>0$,
		\begin{equation*}
			\|u_\varepsilon^\infty-u^\infty\|_{C(\mathbb{S}^{N-1})}\leq C_1\varepsilon^{1/2}\left(\|u^i\|_{H^1(B_r\setminus\overline{\Omega})}+\|f\|_{L^2(B_{r_0}\setminus\overline\Omega)}\right).
		\end{equation*}
		Here, $C_1$ depends only on $q, \gamma, k,\eta_0,\tau_0,\varepsilon_0, G$, $\Omega$, and $B_r$, but completely independent of $\varepsilon$.
	\end{thm}

	\begin{remark}\label{Re:2}

	This theorem offers a quantitative characterization of the degree of approximation between the far-field patterns associated with the scattering problems \eqref{equ:1.2} and \eqref{equ:1.5}.  It rigorously demonstrates the existence of an $\varepsilon^{1/2}$-realization of obstacles $G$ as defined in Definition \ref{def:1.1}. For the inverse problem \eqref{equ:1.4}, a standard approach involves the following optimization formulation:
\begin{equation*}\min_{(\hat{G}\oplus(\Omega\backslash\overline{\hat{G}});\gamma,q)\in\mathcal{S}}\left\|u^\infty\left(\hat{x};u^i, f, (\hat{G}\oplus(\Omega\backslash\overline{\hat{G}});\gamma,q)\right)-\mathcal{M}\left(\hat{x};u^i, f, (G\oplus(\Omega\backslash\overline{G});\gamma,q)\right)\right\|_{C(\mathbb{S}^{N-1})},
\end{equation*}
where $\mathcal{S}$ denotes the a priori class of admissible scatterers and $\mathcal{M}$ represents the measured far-field data. Based on Definition~\ref{def:1.1} and Theorem~\ref{thm:1.1}, the problem can be reformulated into a minimization framework as follows:
\begin{equation}\min_{(\Omega;\gamma_{\varepsilon(x)},q_{\varepsilon(x)})\in\mathcal{Q}}\left\|u^\infty\left(\hat{x};u^i, f, (\Omega;\gamma_{\varepsilon(x)},q_{\varepsilon(x)})\right)-u^\infty\left(\hat{x};u^i, f, (G\oplus(\Omega\backslash\overline{G});\gamma,q)\right)\right\|_{C(\mathbb{S}^{N-1})},\label{equ:1.12}
\end{equation}	where $\mathcal{Q}$ represents the class of admissible scatterers. Notably,
 $(\Omega;\gamma_{\varepsilon(x)},q_{\varepsilon(x)})$ serves as an asymptotic global minimizer for \eqref{equ:1.12}. Furthermore, the significant differences in the asymptotic behaviors of $\gamma_{\varepsilon(x)}$ and $q_{\varepsilon(x)}$ within $G$ and $\Omega\backslash\overline{G}$ allow us to locate and determine these obstacles. For a more detailed discussion, please refer to Subsection \ref{sub:tech}.

We emphasize that the parameter $\varepsilon$ in Theorem \ref{thm:1.1} provides a quantitative characterization for effectively approximating the sound-soft obstacles $G$ as acoustic mediums. This parameter can be chosen prior to implementing optimization-based inverse schemes. Typically, $\varepsilon$ is a small positive constant, but not arbitrarily small, meaning we do not assume $\varepsilon \to 0^+$. The physical parameters $\gamma_\varepsilon(x)$ and $q_\varepsilon(x)$, as defined in \eqref{eq:gamma}, characterize the approximate acoustic media for $G$ and are related to $1/\varepsilon$. Consequently, $\gamma_\varepsilon(x)$ and $q_\varepsilon(x)$ remain bounded but may take large values. In contrast, the physical parameters $\gamma(x)$ and $q(x)$, defined in \eqref{eq:qq1} and characterizing the media $\Omega \setminus \overline{G}$, are typically smaller than $\gamma_\varepsilon(x)$ and $q_\varepsilon(x)$. This distinction enables the identification of the boundary of $G$ through optimization-based inverse schemes. Based on the effective medium theory, we will analyze the behavior of the parameters $\gamma(x)$ and $q(x)$ in the region $\Omega \setminus \overline{G}$. By employing various optimization methods, we aim to develop a numerical reconstruction technique that simultaneously determines both the shape and position of the sound-soft obstacles $G$, which is the focus of our future work.

\end{remark}

\begin{thm}\label{thm:1.2}
Let $u_\varepsilon\in H_{loc}^1(\mathbb{R}^N)$ be the solution to (\ref{equ:1.5}). Then there exists a small $\varepsilon_{0}>0$ such that the following estimate holds for $\varepsilon<\varepsilon_{0}:$
\begin{equation*}
\left\|u_{\varepsilon}^{+}\right\|_{H^{1/2}(\partial G)}\leq C_2\varepsilon^{1/2}\left(\|u^i\|_{H^1(B_r\setminus\overline{\Omega})}+\|f\|_{L^2(B_{r_0}\setminus\overline{\Omega})}\right),
\end{equation*}
where $B_{r_0}$ and $B_{r}$ are defined in Theorem \ref{thm:1.1}, $C_2$ depends only on $q,k,\eta_0,\tau_0,\varepsilon_0, G $, $ \Omega$, and $B_r$, and is completely independent of $\varepsilon$.
\end{thm}
\begin{remark} 
Theorem  \ref{thm:1.2} indicates that when the obstacle $G$ is regarded as a medium, the total wave field on $\partial G$ nearly vanishes, namely, the effective medium theory is feasible when we consider $G$ as a sound-soft obstacle. The estimates outlined in Theorem \ref{thm:1.1} and Theorem \ref{thm:1.2} are optimal. In Section \ref{sec:5}, we present a concrete example that rigorously demonstrates the sharpness of these effective approximations, thereby providing strong support for our theoretical conclusions.
\end{remark}

\subsection{Existing results, technical developments,  and discussions}\label{sub:tech}
As previously discussed, this study is primarily motivated by the investigation of the inverse scattering problem (\ref{equ:1.4}). The inverse scattering problem is inherently nonlinear and ill-posed. Extensive research has been conducted on scenarios involving background mediums devoid of obstacles and inhomogeneous mediums with penetrable obstacles; see \cite{Cakoni2003,Cakoni2004,DFLW2022,DGT2025,Hahenr93,Haher1998, ref6,ref99,SU1987} for further details. While the presence of an impenetrable obstacle introduces additional complexities, significant progress has been made in this area. Notable studies include Isakov’s singular source method for the unique recovery of a purely sound-hard obstacle at a fixed wavenumber \cite{Kk1993} and research on the unique recovery of an embedded obstacle, either purely sound-soft or sound-hard,  using finitely many wavenumbers and incident directions \cite{LZ2006}. Recently, it was demonstrated in \cite{LL17} that simultaneous reconstruction of an embedded obstacle and its surrounding inhomogeneous medium can be successfully achieved using multi-frequency far-field data.
Numerical investigations have also contributed to advancements in inverse obstacle scattering problems. For example, a sampling method was developed in \cite{LL17} to reconstruct the shape and location of buried targets, as well as the support of surrounding inhomogeneities. Additionally, a Newton method was proposed in \cite{Z} for reconstructing the interface and buried obstacle without prior knowledge of the boundary type, assuming an isotropic medium with a specified constant internal refractive index.

 In this paper, we address a time-harmonic inverse scattering problem involving multiple sound-soft obstacles embedded within an inhomogeneous anisotropic medium, surrounded by a homogeneous, non-damped medium. Currently, existing studies, such as those in \cite{ref14,ref13}, primarily focus on inverse scattering problems within simply connected domains. In contrast, our work considers a more challenging mathematical model involving multiple embedded impenetrable obstacles, which introduces multiple boundaries and results in a multi-connected region with complex topologies. Furthermore, most existing literature, such as \cite{LL17, YZ}, focuses on cases involving a single simply connected impenetrable obstacle situated in an isotropic medium. Research on inverse scattering from impenetrable obstacles in anisotropic mediums remains quite limited, with only two notable references identified \cite{KW2021, LZZ2015}. Our study, however, addresses the more complex scenario of multiple impenetrable obstacles within an anisotropic medium. The combination of these topological structures and the anisotropic nature of the medium significantly increases the difficulty of reconstructing such complex scatterers.

 To address these challenges, we adopt the effective medium theory outlined above to investigate the inverse scattering problem \eqref{equ:1.4}. Specifically, we demonstrate that these sound-soft obstacles can be effectively approximated by an isotropic and lossy medium characterized by the carefully chosen physical parameters $\gamma_{\varepsilon(x)}$ and $q_{\varepsilon(x)}$. In this context, the total wave field corresponding to the medium scattering problem \eqref{equ:1.5} exhibits decay properties that relate to these specified material parameters at the boundaries of the obstacles. Furthermore, we can provide rigorous estimates to validate this approximation and present a concrete example illustrating the optimality of the estimates outlined in Theorem \ref{thm:1.1} and Theorem \ref{thm:1.2}. As discussed in Remark \ref{Re:2}, the optimization problem \eqref{equ:1.12} indicates that $(\Omega;\gamma_{\varepsilon(x)},q_{\varepsilon(x)})$ serves as an asymptotic global minimizer; thus, we can approximately locate the topological structure of the underlying scatterer 
$(\hat{G}\oplus(\Omega\backslash\overline{\hat{G}});\gamma,q)$. Indeed, it can be observed that $(\Omega;\gamma_{\varepsilon(x)},q_{\varepsilon(x)})$ regard as an effective $\varepsilon^{1/2}$-realization of $G\oplus(\Omega\setminus\overline{G})$, possesses distinct asymptotic behavior for the physical parameters $\gamma_{\varepsilon(x)}$ and $q_{\varepsilon(x)}$ on either side of $\partial G$. Specifically, the parameters $\gamma_{\varepsilon(x)}$ and $q_{\varepsilon(x)}$ become significantly large within the subregion $G$, while remaining bounded in 
$\Omega \setminus \overline{G}$. This remarkable difference in parameter behavior enables us to determine the locations and shapes of these obstacles. Our results in this paper, combined with the extensions presented in \cite{Liu2012} concerning sound-hard obstacles buried in an isotropic and inhomogeneous medium, can identify the types of obstacles and determine the locations of the complex topological structure without prior knowledge of the obstacles. In comparison to the study concerning the sound-hard obstacle in \cite{Liu2012}, we introduce a new Hilbert space to address the challenges presented by non-homogeneous Dirichlet boundary conditions.


\par The remaining sections of this paper are organized as follows: In Section \ref{sec:3}, we introduce several key lemmas that are essential for establishing Theorems \ref{thm:1.1} and \ref{thm:1.2}. Section \ref{sec:4} presents detailed proofs of these theorems. In Section \ref{sec:5}, we examine a concrete example where $G$ degenerates into a single sphere, thereby demonstrating the optimality of the estimates provided in Theorems \ref{thm:1.1} and \ref{thm:1.2}. Finally, the appendix offers a critical discussion concerning the well-posedness of the problem \eqref{equ:1.5}.

\section{Auxiliary results }\label{sec:3}
This section focuses on presenting auxiliary results that will facilitate the proofs of Theorems \ref{thm:1.1} and \ref{thm:1.2}. 
\subsection{The well-posedness and energy estimation on the scattering problem (\ref{equ:2.1}) }
This subsection introduces several lemmas that will be utilized in the subsequent proofs of the main theorems. First, we need to establish the well-posedness of a more general scattering problem \eqref{equ:2.1}. Compare with \eqref{equ:1.2}, there are more flexible choice of boundary values in \eqref{equ:2.1}.  Let $(v,u^s)\in H^1(\Omega\setminus\overline{G})\times H_{loc}^1(\mathbb{R}^N\backslash\overline{\Omega})$ satisfy the following PDE system 
\begin{equation}
\begin{cases}\nabla\cdot(A\nabla v)+k^{2}q(x)v=0&\text{in}   ~~\Omega\setminus\overline{G},\vspace*{1mm}\\\Delta u^{s}+k^{2}u^{s}=f&\text{in}~~\mathbb{R}^{N}\setminus\overline{\Omega},\vspace*{1mm}\\ v=p\in H^{1/2}(\partial G)&\text{on} ~~\partial G,\vspace*{1mm}\\ v-u^{s}=h_{1}\in H^{1/2}(\partial\Omega)&\text{on}~~\partial\Omega,\vspace*{1mm}\\\frac{\partial v}{\partial\nu_{A}}-\frac{\partial u^{s}}{\partial\nu}=h_{2}\in H^{-1/2}(\partial\Omega)&\text{on}~~\partial\Omega,\vspace*{1mm}\\\lim\limits_{|x|\to+\infty}|x|^{(N-1)/2}\left\{\frac{\partial u^{s}}{\partial|x|}-iku^{s}\right\}=0,\label{equ:2.1}\end{cases}
\end{equation}
where $p\in H^{1/2}(\partial G)$, $h_{1}\in H^{1/2}(\partial\Omega)$, $h_{2}\in H^{-1/2}(\partial\Omega)$, and $f$ satisfies $\mathrm{supp}(f)\subset B_{r_{0}}\backslash\overline{\Omega}\subset B_{r}\backslash\overline{\Omega}$.
\par Furthermore, we employ the  Dirichlet-to-Neumann (DtN) map to transform (\ref{equ:2.1}) into a bounded domain problem (\ref{equ:2.2}), which can be defined as follows: let $(v_{1},u_{1})\in H^1({\Omega\setminus\overline{G}})\times H^1(B_r\backslash\overline{\Omega})$ be such that 
\begin{equation} 	\begin{cases}\nabla\cdot(A\nabla v_{1})+k^{2}q(x)v_{1}=0&\text{in}   ~~\Omega\setminus\overline{G},\vspace*{1mm}\\\Delta u_{1}+k^{2}u_{1}=f&\text{in}~~\hspace*{-1mm}\ B_{r}\backslash\overline{\Omega},\vspace*{1mm}\\ v_1=p&\text{on} ~~\partial G,\\v_{1}-u_{1}=h_{1}&\text{on} ~~\partial\Omega,\vspace*{1mm}\\\frac{\partial v_{1}}{\partial\nu_{A}}-\frac{\partial u_{1}}{\partial\nu}=h_{2}&\text{on} ~~\partial\Omega,\vspace*{1.3mm}\\\frac{\partial u_{1}}{\partial\nu}=\Lambda u_{1}&\text{on} ~~\partial B_{r}. \label{equ:2.2}\end{cases}	\end{equation}
Here, $p$, $h_1$, $h_2$, and $f$ are as defined earlier. The operator $\Lambda:H^{1/2}(\partial B_r)\to H^{-1/2}(\partial B_r)$ is the DtN map. Thus, ${\Lambda}: \widetilde{M}\longmapsto\frac{\partial\widetilde{T}}{\partial\nu}$, where $\widetilde{T}$ is a radiating solution in $H_{loc}^{1}(\mathbb{R}^{N}\backslash\overline{B}_{r})$ satisfying the Helmholtz equation
\begin{center}
$\begin{cases}\Delta\widetilde{T}+k^2 \widetilde{T}=0~~&\text{in}~~\mathbb{R}^N\backslash\overline{B_r},\\[1.6ex]
\widetilde{T}=\widetilde{M}\in H^{1/2}(\partial B_r)~~&\text{on}~~\partial B_r.\end{cases}$
\end{center}
With all these preparations, we now proceed to present the following lemmas.

\begin{lem}\label{lem:2}
The scattering problems (\ref{equ:2.1}) and (\ref{equ:2.2}) are equivalent.
\end{lem}	
\begin{proof}
By applying the definition of $\Lambda$, it is straightforward to observe that if $(v, u^{s})$ solves the scattering problem (\ref{equ:2.1}), then $(v,u^{s})|_{B_{r}\backslash\overline{G}}$ solves the scattering problem (\ref{equ:2.2}). On the other hand, suppose $(v_1,u_1)$ solves the truncated system (\ref{equ:2.2}). By applying Green's integral representation, we obtain
\begin{equation}\begin{aligned}
	u_1(x)=&\int_{\partial B_r}\Lambda u_1\cdot\Phi(x,y)-u_1(y)\cdot\frac{\partial\Phi(x,y)}{\partial\nu(y)}\mathrm{d}s-\int_{\partial\Omega}\frac{\partial u_1(y)}{\partial\nu(y)}\cdot\Phi(x,y)-u_1(y)\cdot\frac{\partial\Phi(x,y)}{\partial\nu(y)}\mathrm{d}s\\&-\int_{B_r\setminus\overline{\Omega}}\Phi(x,y)\cdot f(y)\mathrm{d}y,\label{equ:2.7}
\end{aligned}
\end{equation}
where  $\frac{\partial u_1}{\partial\nu} = \Lambda u_1$ on $\partial B_r$, and $\Phi(x,y)$ is the fundamental solution to the Helmholtz equation, given by
\begin{equation*}\Phi(x,y)=\frac{i}{4}\left(\frac{k}{2\pi|x-y|}
\right)^{(N-2)/2}H_{(N-2)/2}^{(1)}(k|x-y|),\end{equation*}
with $N = 2,3$, $x \in B_r\setminus\overline{\Omega}$, and $H_n^{(1)}(t)$ representing the $n$-th order Hankel function of the first kind. Since $\Lambda$ and $\Phi(x,y)$ are defined to satisfy the Sommerfeld radiation condition, it follows that
\begin{equation}\label{equ:2.8}
\int_{\partial B_r}\Big\{\frac{\partial\Phi(x,y)}{\partial\nu(y)}\cdot u_1(y)-\Phi(x,y)\cdot\Lambda u_1(y)\Big\}\mathrm{d}s=0.
\end{equation}
Substituting (\ref{equ:2.8}) into (\ref{equ:2.7}), we obtain
\begin{equation*}
\begin{aligned}u_1(x)=\int_{\partial\Omega}\Big\{\frac{\partial\Phi(x,y)}{\partial\nu(y)}\cdot u_1(y)-\Phi(x,y)\cdot\frac{\partial u_1(y)}{\partial\nu(y)}\Big\} \mathrm{d}s-\int_{B_r\setminus\overline{\Omega}}\Phi(x,y)\cdot f(y) \mathrm{d}y.
\end{aligned}\end{equation*}
Clearly, $u_1$ can be extended to a function in $H_{loc}^1(\mathbb{R}^N\backslash\overline{\Omega})$, and we continue to denote it by $u_1$. Since each variable of $\Phi(x,y)$ satisfies the Sommerfeld radiation condition, the extended function $u_1 \in H_{loc}^1(\mathbb{R}^N\backslash\overline{\Omega})$ also satisfies this condition. Thus, $(v_1, u_1)$ solves the problem (\ref{equ:2.1}).	
\end{proof}
The following lemma shows that the PDE system (\ref{equ:2.1}) is well-posed and admits the energy estimate \eqref{equ:2.10}.
\begin{lem}\label{lem:2.3}
There exists a unique solution $(v,u^s)\in H^1(\Omega\setminus\overline{G})\times H_{loc}^1(\mathbb{R}^N\backslash\overline{\Omega})$ to (\ref{equ:2.1}), and the solution satisfies 
\begin{equation}\begin{aligned}\|v\|_{H^{1}(\Omega\setminus\overline{G})}+\|u^{s}\|_{H^{1}(B_{r}\setminus\overline{\Omega})}\leq & C\Big(\|p\|_{H^{1/2}(\partial G)}+\|h_{1}\|_{H^{1/2}(\partial\Omega)}+\|h_{2}\|_{H^{-1/2}(\partial\Omega)}\\
&\hspace{4.5cm}+\|f\|_{L^{2}(B_{r_{0}}\setminus\overline{\Omega})}\Big),\label{equ:2.10}\end{aligned}\end{equation} where $r>r_0$,  $C$ is positive depending only on $\gamma$, $q$, $k$, $\Omega$, $G$, and $B_r$, but independent of $p$, $h_1$, $h_2$, $f$.
\end{lem}
\begin{proof} 
\medskip  \noindent {\it Step I:} Prove that (\ref{equ:2.1}) has at most one solution.

\medskip \noindent Let $p = 0$, $h_1 = 0$, $h_2 = 0$, and $f = 0$. Subsequently, the first equation in (\ref{equ:2.1}) is multiplied by $\overline{v}$ and integrated over $\Omega\backslash\overline{G}$, yielding
\begin{center}
$\int_{\Omega\setminus\overline{G}}\,\nabla\cdot(A\nabla v)\,\overline{v}+k^{2}q(x)v^{2}\mathrm{d}x=0.$
\end{center}
Applying Green's first formula, we obtain
\begin{equation}\label{equ:2.3}
\begin{aligned}
	\int_{\partial\Omega}\overline{v}\,\nu\cdot A\nabla v\,\mathrm{d}s-\int_{\partial G}\overline{v}\,\nu\cdot A\nabla v\,\mathrm{d}s-\int_{\Omega\setminus\overline{G}}\nabla\overline{v} \cdot A\nabla v\,\mathrm{d}x+\int_{\Omega\setminus\overline{G}}k^{2}q(x)v^{2}\mathrm{d}x=0.
	\end{aligned}\end{equation}
	Combining this with the boundary conditions on $\partial G$, we obtain
	\begin{align}
\int_{\partial\Omega}\overline{v}\,\nu\cdot A\nabla v\,\mathrm{d}s-\int_{\Omega\setminus\overline{G}}\nabla\overline{v} \cdot A\nabla v\,\mathrm{d}x+\int_{\Omega\setminus\overline{G}}k^{2}q(x)v^{2}\,\mathrm{d}x=0.
\label{equ:2.4}
\end{align}
For the second equation, we multiply both sides of the equation by $\overline{u^s}$ and integrate the resulting expression over $B_{r}\backslash\overline{\Omega}$. Then it follows that
\begin{center}
\scalebox{1.04}{$\int_{B_{r}\backslash\overline{\Omega}}\,\Delta u^{s}\,\overline{u^{s}}+k^{2}|u^{s}|^{2}\,\mathrm{d}x=0.$}
\end{center}
Applying Green's first formula again, we derive
\begin{align}
\int_{\partial B_r}\frac{\partial u^s}{\partial\nu}\overline{u^s}\mathrm{d}s-\int_{\partial\Omega}\frac{\partial u^s}{\partial\nu}\overline{u^s}\mathrm{d}s-\int_{B_{r}\backslash\overline{\Omega}} |\nabla u^s|^2\,\mathrm{d}x+\int_{B_{r}\backslash\overline{\Omega}}k^2|u^s|^2\,\mathrm{d}x=0.\label{equ:2.5}
\end{align}
Adding \eqref{equ:2.3}, (\ref{equ:2.4}) and (\ref{equ:2.5}), along with the boundary conditions on $\partial\Omega$, we obtain
\begin{equation}
\begin{aligned}
	&\int_{\Omega\setminus\overline{G}}k^{2}q(x)v^{2}\,\mathrm{d}x-\int_{\Omega\setminus\overline{G}}\nabla\overline{v} \cdot A\nabla v\,\mathrm{d}x-\int_{B_{r}\backslash\overline{\Omega}} |\nabla u^s|^2\,\mathrm{d}x+\int_{B_{r}\backslash\overline{\Omega}}k^2|u^s|^2\,\mathrm{d}x	\\
&\hspace{9cm} +\int_{\partial B_r}\frac{\partial u^s}{\partial\nu}\overline{u^s}\,\mathrm{d}s=0.\label{equ:2.6}
\end{aligned}
\end{equation}
Taking the imaginary part of (\ref{equ:2.6}), we deduce that
\begin{center}
$\Im\int_{\partial B_r}\frac{\partial{u}^s}{\partial\nu}\overline{{u}^s}\,\mathrm{d}s(x)=-k^{2}\Im\int_{\Omega\setminus\overline{G}}q(x)v^{2}\,\mathrm{d}x+\Im\int_{\Omega\setminus\overline{G}}\nabla\overline{v} \cdot A\nabla v\,\mathrm{d}x\le0.$
\end{center}
From Rellich's lemma and the unique continuation principle, it follows that $u^s = 0$ in $\mathbb{R}^{N}\backslash\overline{B}_{r}$. Consequently, $u^s = 0$ in $B_{r}\backslash\overline{\Omega}$ and $v = 0$ in $\Omega\setminus\overline{G}$. Thus, the uniqueness of the solution to (\ref{equ:2.1}) is established.

\medskip  \noindent {\it Step II:} Verify the existence of a solution to \eqref{equ:2.1} and demonstrate that it satisfies the estimate \eqref{equ:2.10}.

\medskip  \noindent
By Lemma \ref{lem:2}, we know that the problems \eqref{equ:2.1} and \eqref{equ:2.2} are equivalent. Therefore, we will focus on establishing the existence of a solution to \eqref{equ:2.2} that meets the estimate \eqref{equ:2.10}. For simplicity, we assume that $k^2$  is not a Dirichlet eigenvalue of $-\Delta$ in $B_r\backslash\overline{\Omega}$. Define $w(x)=v_1(x)\mathrm{~in~}\Omega\backslash\overline{G} $ and $w(x)=u_1(x)+\widetilde{v}(x)$ in $B_{r}\backslash\overline{\Omega}$. Consequently, the PDE system (\ref{equ:2.2}) can be reformulated as follows
\begin{equation}
\begin{cases}
	\nabla\cdot(A\nabla w)+k^2qw=0&\text{in}~\Omega\setminus\overline{G},\vspace*{1mm}\\
	\Delta w+k^2w=f&\text{in}~B_r\backslash\overline{\Omega},\vspace*{1mm}\\
	w=p&\text{on}~{\partial{G}}, \vspace*{1mm}\\
	w^-=w^+$,\quad$\frac{\partial w^-}{\partial\nu_{A}}=\frac{\partial w^+}{\partial\nu}+h_2-\frac{\partial\widetilde{v}}{\partial\nu}&\text{on}~\partial\Omega,\vspace*{1mm}\\
	\frac{\partial w}{\partial\nu}=\Lambda w+\frac{\partial\widetilde{v}}{\partial\nu}&\text{on}~\partial B_r,\label{equ:2.11}
	\end{cases}\end{equation}
	where $\widetilde{v}$ satisfies
	\begin{equation*}
\begin{cases}
	\Delta\widetilde{v}+k^2\widetilde{v}=0&\text{in~}B_r\backslash\overline{\Omega},\vspace*{1mm}\\
	\widetilde{v}=h_1&\text{on }\partial\Omega,\vspace*{1mm}\\
	\widetilde{v}=0&\text{on}~\partial B_r.
\end{cases}
\end{equation*}
It is evident that $\widetilde{v}$ is unique, and it follows that $\|\tilde{v}\|_{H^{1}(B_{r}\setminus\overline{\Omega})}\leq C'\|h_{1}\|_{H^{1/2}(\partial\Omega)}$, where $C'$ is a positive constant independent on $\varepsilon$. Subsequently, we introduce a bounded operator $\Lambda_0:H^{1/2}(\partial B_r)\longrightarrow H^{-1/2}(\partial B_r)$, which maps $\phi\text{ to }\frac{\partial\widetilde{w}}{\partial\nu}\Big|_{\partial B_r}.$ Here, $\widetilde{w}\in H_{loc}^{1}(\mathbb{R}^{N}\backslash\overline{{B_{r}}})$ is the unique solution to the following system
\begin{equation*}
\begin{cases}
	\Delta\widetilde w=0&\text{in}~~\mathbb R^N\backslash\overline{B_r},
	\vspace*{3mm}\\
	\widetilde w=\phi\in H^{1/2}(\partial B_r)&\text{on}~~\partial B_r,
\end{cases}
\end{equation*}
and satisfies the decay property at infinity, specifically $\widetilde w=\mathcal{O}(\operatorname{log}|x|)$ for $ N=2$ and $\widetilde w=\mathcal{O}(|x|^{-1})$ for $N=3$, as $|x|\to+\infty$. The operator $\Lambda_0$ possesses the following properties \cite{ref4,ref6}:
\begin{center}
$\left\{\begin{array}{l} -\int_{\partial B_r}\overline{\phi} ~\Lambda_0~\phi ds\geq0 \quad\quad\phi\in H^{1/2}(\partial B_r),\vspace*{5mm}\\\Lambda-\Lambda_0:H^{1/2}(\partial B_r)\to H^{-1/2}(\partial B_r) ~\text{is compact}.
\end{array}\right.$
\end{center}
Next, we introduce a new space $\mathcal{V}$ as follows:
\begin{center}
$\mathcal{V} :=\left \{ w\in H^1(B_r\setminus\overline{G}); w=0 ~\text{on}~\partial G \right \}.$\end{center}Let $w_0\in H^1(B_r\setminus\overline{G})$	
be the solution to the PDE system\begin{equation}
\begin{cases}
	-\Delta{w_0}-k^2{w_0}=0&\text{in~}B_r\backslash\overline{G},\vspace*{1mm}\\
	{w_0}=p&\text{on }\partial G, 	\vspace*{1mm}\\
	{w_0}=0&\text{on}~\partial B_r.\label{equ:8.8}
\end{cases}
\end{equation}
Based on the well-posedness of (\ref{equ:8.8}), we can conclude that $\|w_0\|_{H^1(B_r\backslash\overline{G})}\leq \widetilde{C}\|p\|_{H^{1/2}(\partial G)}$, where $\widetilde{C}$ is a positive constant independent on $\varepsilon$. Then for any $\varphi\in \mathcal{V}$, we can easily obtain the variational formulation of \eqref{equ:2.11}: find $w\in{H^1(B_r\backslash\overline{G})}$ such that
\begin{align*} a_1(w-w_0,\varphi)+a_2(w-w_0,\varphi)=\mathcal{F}(\varphi),
\end{align*}
where the bilinear forms $a_1(\cdot,\cdot)$ and $a_2(\cdot,\cdot)$ and the linear functional $\mathcal{F}(\varphi)$ are defined by
 \begin{align*}
a_1(w-w_0,\varphi)&:=\int_{\Omega\setminus\overline{G}}\nabla\bar{\varphi}\cdot A\nabla(w-w_0)\,\mathrm{d}x+\int_{B_r\setminus\overline{G}} k^2(w-w_0)\bar{\varphi}\,\mathrm{d}x\\
&\quad\,\,-\int_{\partial B_r}\Lambda_0(w-w_0)\bar{\varphi}\,\mathrm{d}s\mathrm{d}x+\int_{B_r\setminus\overline{\Omega}}\nabla (w-w_0)\cdot\nabla\bar{\varphi},\nonumber\\
a_2(w-w_0,\varphi)&:=-\int_{\Omega\setminus\overline{G}}k^2(q+1)(w-w_0)\bar{\varphi}\,\mathrm{d}x-2\int_{B_{r}\setminus\overline{\Omega}} k^{2}(w-w_0)\bar{\varphi}\,\mathrm{d}x\\
&\quad\,\, -\int_{\partial B_{r}} (\Lambda-\Lambda_{0})(w-w_0)\bar{\varphi}\,\mathrm{d}s\nonumber,
\\
\mathcal{F}(\varphi)&:=\int_{\partial\Omega}(h_{2}-\frac{\partial\tilde{v}}{\partial\nu})\bar{\varphi}\,\mathrm{d}s+\int_{\partial B_{r}}\frac{\partial\tilde{v}}{\partial\nu}\bar{\varphi}\,\mathrm{d}s+\int_{\Omega\setminus\overline{G}}\nabla w_0\cdot\nabla\bar{\varphi}\,\mathrm{d}x\\
&\quad\,\,-\int_{B_r\setminus\overline{\Omega}}f\bar{\varphi}\,\mathrm{d}x-\int_{\Omega\setminus\overline{G}}\nabla\bar{\varphi}\cdot A\nabla w_0\,\mathrm{d}x
+\int_{\Omega\setminus\overline{G}}k^2(q-1)w_0\,\bar{\varphi}\,\mathrm{d}x.
\end{align*}
Since $\Lambda$ is a bounded operator from $H^{1/2}(\partial B_r)$ to $H^{-1/2}(\partial B_r)$, $\mathcal{F}$ is a bounded conjugate linear functional on $\mathcal{V}$ and both $a_1(\cdot,\cdot)$ and $a_2(\cdot,\cdot)$ are continuous on $\mathcal{V}\times \mathcal{V}$.
Clearly, from (\ref{equ:10.1}), it is evident that for any $\phi, \varphi \in \mathcal{V}$, the function $a_1(\cdot,\cdot)$ satisfies both the conditions of coercivity and positive definiteness, as follows:
\begin{equation*}
\begin{aligned}|a_1(\phi,\varphi)|\leq C_1\|\phi\|_{H^1(B_r\backslash\overline{G})}\|\varphi\|_{H^1(B_r\backslash\overline{G})}~~\mathrm{and}\quad a_1(\varphi,\varphi)\geq C_2\|\varphi\|_{H^1(B_r\backslash\overline{G})}^2,
\end{aligned}
\end{equation*}
where $C_1$ and $C_2$ are positive constants independent on $\varepsilon$. According to the Lax-Milgram lemma, there exists a bounded operator $\mathcal{S}$: $\mathcal{V}\to \mathcal{V}$ such that	
\begin{equation*} 
\begin{aligned} a_1(w-w_0,\varphi)=\left \langle \mathcal{S}(w-w_0),\varphi  \right \rangle,
\end{aligned}\end{equation*}
where the notation $\left \langle\cdot,\cdot  \right \rangle $  denotes the inner product in $H^1(B_r\backslash\overline{G})$. Thus the inverse of $\mathcal{S}$ is also bounded.  From the expression of the bilinear form $a_2(\cdot,\cdot)$, we introduce two bounded operators $\mathcal{R}_1$ and $\mathcal{R}_2: \mathcal{V} \to \mathcal{V}$ given by
\begin{align}
a_{3}(w-w_0,\varphi)&:=\int_{\Omega\setminus\overline{G}}k^2(q+1)(w-w_0)\bar{\varphi}\,\mathrm{d}x+2\int_{B_{r}\setminus\overline{\Omega}} k^{2}(w-w_0)\bar{\varphi}\,\mathrm{d}x
=\left \langle\mathcal{R}_{1}(w-w_0),\varphi\right \rangle,\label{equ:2.20}\\
a_4(w-w_0,\varphi)&:=\int_{\partial B_r}(\Lambda-\Lambda_0)(w-w_0)\bar{\varphi}\,\mathrm{d}s=\left \langle\mathcal{R}_2(w-w_0),\varphi\right \rangle.
\label{equ:2.21}
\end{align}

Next, we aim to prove that the operators $\mathcal{R}_{1}$ and $\mathcal{R}_{2}$ are compact. Let $\{Q_n\}_{n\in\mathbb{N}}$ be a bounded sequence in $H^1(B_r\backslash\overline{G}),$ and we can assume that $\|Q_n\|_{H^1(B_r\backslash\overline{G})}\leq{M}$ with $Q_n \rightharpoonup Q_0$ in $H^1(B_r\backslash\overline{G})$, where $M$ is a positive constant. Since $H^1(B_r\backslash\overline{G})\hookrightarrow L^2(B_r\backslash\overline{G})$, it follows that $Q_n\to Q_0$ in $L^2(B_r\backslash\overline{G})$. By (\ref{equ:2.20}), we can derive that
\begin{equation*}\begin{aligned} a_3(Q_n-Q_0,\varphi)=\left \langle\mathcal{R}_1(Q_n-Q_0),\varphi  \right \rangle.
\end{aligned}	\end{equation*}
Let $\varphi=\mathcal{R}_1(Q_n-Q_0)$, then, from (\ref{equ:2.20}), we can easily obtain
\begin{center}			$\|\mathcal{R}_1(Q_n-Q_0)\|^2_{H^1(B_r\backslash\overline{G})}\leq2Mk^2\max\{\|q+1\|_{L^\infty(\Omega\setminus\overline{G})},2\}\|\mathcal{R}_1\|\|Q_n-Q_0\|_{L^2(B_r\backslash\overline{G})}\to0.$
\end{center} 
Therefore, we establish the compactness of $\mathcal{R}_1$. Similarly, we can prove the compactness of $\mathcal{R}_2$. Indeed, considering $Q_n \rightharpoonup Q_0$ in $H^1(B_r\backslash\overline{G})$ and $\|Q_n\|_{H^1(B_r\backslash\overline{G})}\leq M$, by virtue of the trace theorem, we deduce that $Q_n|_{\partial B_r} \rightharpoonup Q_0|_{\partial B_r}$ in $H^{1/2}(\partial B_r)$. Moreover, since the operator $\Lambda-\Lambda_{0}:H^{1/2}(\partial B_{r})\to H^{-1/2}(\partial B_{r})$ is compact, we have $(\Lambda-\Lambda_0)Q_n\to (\Lambda-\Lambda_0)	Q_0$ in $H^{-1/2}(\partial B_r)$. From (\ref{equ:2.21}), we find that
\begin{center}
$a_4(Q_n-Q_0,\varphi)=\left \langle\mathcal{R}_2(Q_n-Q_0),\varphi\right \rangle.$
\end{center}
Then, let $\varphi=\mathcal{R}_2(Q_n-Q_0)$, and combining this with (\ref{equ:2.21}), we have
\begin{equation*}
\begin{aligned}
	\|\mathcal{R}_{2}(Q_{n}-Q_{0})\|^2_{H^{1}(B_r\backslash\overline{G})}&\leq\|(\Lambda-\Lambda_{0})(Q_{n}-Q_{0})\|_{H^{-1/2}(\partial B_{r})}\|\mathcal{R}_{2}(Q_{n}-Q_{0})\|_{H^{1/2}(\partial B_{r})}\\&\leq C_{3}M\|(\Lambda-\Lambda_{0})(Q_{n}-Q_{0})\|_{H^{-1/2}(\partial B_{r})}\|\mathcal{R}_{2}\|\to0.
    \end{aligned}\end{equation*}
	Clearly, we have $\|\mathcal{R}_{2}(Q_{n}-Q_{0})\|^2_{H^{1}(B_r\backslash\overline{G})}\to0$. Thus, we have proved the compactness of $\mathcal{R}_{2}$ and can conclude that
	\begin{center}	$\left \langle(\mathcal{S}-\mathcal{R}_1-\mathcal{R}_2) (Q_n-Q_0), \varphi\right \rangle=\mathcal{F}(\varphi).$
	\end{center} 
	Since $\mathcal{S}$ is a bounded and invertible operator, and $\mathcal{R}_{1}+\mathcal{R}_{2}$ is a compact operator, we conclude that $\mathcal{S}-(\mathcal{R}_1+\mathcal{R}_2)$ is a Fredholm operator of index zero, as stated by the Riesz representation theorem. By the uniqueness of (\ref{equ:2.1}), $(\mathcal{S}-(\mathcal{R}_1+\mathcal{R}_2))^{-1}$  is bounded. This implies that
	\begin{equation*}\begin{aligned}\left|\mathcal{F}(\varphi)\right|&\leq C_4(\|h_2\|_{H^{-1/2}(\partial\Omega)}+\|h_1\|_{H^{1/2}(\partial\Omega)}+\|f\|_{L^{2}(B_r\setminus\overline{\Omega})})\|\varphi\|_{H^1({B_r\backslash\overline{G}})}.
    \end{aligned}\end{equation*} 
    Directly, we have
	\begin{equation*}\begin{aligned}
	\left\|v\right\|_{H^1({B_r\backslash\overline{G}})}+\left\|u^s\right\|_{H^1(B_r\setminus\overline{\Omega})}\leq & C_5\Big(\|w_0\|_{H^{1}(B_r\setminus\overline{G})}+\|h_2\|_{H^{-1/2}(\partial\Omega)}+\|h_1\|_{H^{1/2}(\partial\Omega)}\\
&\hspace{5cm}+\|f\|_{L^{2}(B_{r_{0}}\setminus\overline{\Omega})}\Big)\\&\leq C_6\Big(\|p\|_{H^{1/2}(\partial G)}+\|h_2\|_{H^{-1/2}(\partial\Omega)}+\|h_1\|_{H^{1/2}(\partial\Omega)}
\\
&\hspace{5cm}+\|f\|_{L^{2}(B_{r_{0}}\setminus\overline{\Omega})}\Big),
	\end{aligned}\end{equation*}
    where $C_3$, $C_4$, $C_5$ and $C_6$ are positive constants independent on $\varepsilon$. Hence, we finish the proof.
\end{proof}
\subsection{Energy estimations  on the effective medium scattering system (\ref{equ:1.5})}
In this subsection, we aim to demonstrate the existence of two significant energy estimations for the solution to the effective medium scattering system (\ref{equ:1.5}), which is associated with the time-harmonic incident wave.
\begin{lem}  \label{lem:2.4}                              
Let $u_{\varepsilon}\in H_{loc}^{1}(\mathbb{R}^{N})$  be the solution of (\ref{equ:1.5}). Then, there exists a small $\varepsilon_{0}>0$ such that for all $\varepsilon<\varepsilon_{0}$, we can obtain the following estimates:
\begin{align}
&\|u_{\varepsilon}\|_{H^{1}({B_r\backslash\overline{G}})}\leq C_{1}(\|f\|_{L^{2}(B_{r_{0}}\setminus\overline{\Omega})}+\|u^{i}\|_{H^{1}(B_{r}\setminus\overline{\Omega})})\label{equ:2.26},\\
& \|u_{\varepsilon}\|_{H^1(G)}\leq C_2\,\varepsilon^\frac{1}{2} (\|f\|_{L^2(B_{r_0}\setminus\overline{\Omega})}+\|u^i\|_{H^1(B_r\setminus\overline{\Omega})}),\label{equ:2.27}
\end{align}
where $C_1$ and $C_2$ are positive constants independent on $\varepsilon$. 
\end{lem}
\begin{proof}
By multiplying $\bar{u}_{\varepsilon}$ with the first equation of (\ref{equ:1.5}) and integrating over $G$, we can derive that	
\begin{center}
$\int_{G}\,\nabla\cdot(\varepsilon^{-1} \nabla {u}_{\varepsilon})\,\bar{u}_{\varepsilon}+k^{2}(\eta_0+i\varepsilon^{-1}\tau_0)|{u}_{\varepsilon}|^{2}\,\mathrm{d}x=0$.
\end{center} 
By applying  Green's first formula, we have
\begin{center}
$\int_{\partial G}\,\varepsilon^{-1}\frac{\partial{u}^{-}_{\varepsilon}}{\partial v}\,\bar{u}_{\varepsilon}ds-\int_{G}\varepsilon^{-1}\,|\nabla u_{\varepsilon}|^{2}dx+\int_{G}k^{2}(\eta_0+i\varepsilon^{-1}\tau_0)|{u}_{\varepsilon}|^{2}\,\mathrm{d}x=0.  $
\end{center}
According to the fifth equation of (\ref{equ:1.5}), we obtain that
\begin{equation}
\begin{aligned}
	\int_{\partial G}\,\frac{\partial{u}^{+}_{\varepsilon}}{\partial \nu_{A}}\,\bar{u}_{\varepsilon}\,\mathrm{d}s-\int_{G}\,\varepsilon^{-1}|\nabla u_{\varepsilon}|^{2}dx+\int_{G}k^{2}(\eta_0+i\varepsilon^{-1}\tau_0)|{u}_{\varepsilon}|^{2}\,\mathrm{d}x=0.  
	\label{equ:2.28}	
	\end{aligned}\end{equation}
	Utilizing the second equation of (\ref{equ:1.5}), we have 
	\begin{center}
\scalebox{1}{$\int_{\Omega\setminus\overline{G}}\nabla\cdot(A\nabla {u}_{\varepsilon})\bar{u}_{\varepsilon}+k^{2}q(x){u}_{\varepsilon}^{2}\,\mathrm{d}x=0.$}
\end{center}	
Furthermore, we can conclude that
\begin{equation}
\begin{aligned}
	\int_{\partial \Omega}\,\bar{u}_{\varepsilon}~\nu\cdot A\nabla{u}_{\varepsilon}\,\mathrm{d}s-	\int_{\partial G}\,\bar{u}_{\varepsilon}~\nu\cdot A\nabla{u}_{\varepsilon}\,\mathrm{d}s-\int_{\Omega\setminus\overline{G}}\,\nabla\bar u_{\varepsilon}\cdot A\nabla u_{\varepsilon} \,\mathrm{d}x+\int_{\Omega\setminus\bar{G}}k^{2}q(x){u}_{\varepsilon}^{2}\,\mathrm{d}x=0.\label{equ:2.29}
	\end{aligned}\end{equation}
	By adding (\ref{equ:2.28}) to (\ref{equ:2.29}), we have
	\begin{equation}	
\begin{aligned}
	&\int_{\partial \Omega}\,\bar{u}_{\varepsilon}~\nu\cdot A\nabla{u}_{\varepsilon}\,\mathrm{d}s-\int_{G}\,\varepsilon^{-1}|\nabla u_{\varepsilon}|^{2}\,\mathrm{d}x+\int_{G}\,k^{2}(\eta_0+i\varepsilon^{-1}\tau_0)|{u}_{\varepsilon}|^{2}\,\mathrm{d}x\\&-\int_{\Omega\setminus\overline{G}}\,\nabla\bar u_{\varepsilon}\cdot A\nabla u_{\varepsilon} \,\mathrm{d}x+\int_{\Omega\setminus\overline{G}}\,k^{2}q(x){u}_{\varepsilon}^{2}\,\mathrm{d}x=0.
	\label{equ:2.30}
	\end{aligned}\end{equation}
	Multiplying $\bar{u}^{s}_{\varepsilon}$ with the third equation of (\ref{equ:1.5}) and integrating it over $B_r\backslash\overline{\Omega}$, we have
	\begin{center}\scalebox{1.1}{$\int_{B_r\backslash\overline{\Omega}}\,\Delta {u}^{s}_{\varepsilon}\,\overline{u}^{s}_{\varepsilon}+k^{2}{u}^{s}_{\varepsilon}\,\overline{u}^{s}_{\varepsilon}\,\mathrm{d}x=\int_{B_r\backslash\overline{\Omega}}\,f\,\overline{u}^{s}_{\varepsilon}\,\mathrm{d}x.$}
	\end{center}
	Consequently, we can derive 
	\begin{equation}	
\begin{aligned}
	\int_{\partial B_r}\,\frac{\partial {u}^{s}_{\varepsilon}}{\partial v}\,\overline{u}^{s}_{\varepsilon}\,\mathrm{d}s-\int_{\partial \Omega}\,\frac{\partial {u}^{s}_{\varepsilon}}{\partial v}\,\overline{u}^{s}_{\varepsilon}\,\mathrm{d}s-\int_{B_r\setminus\overline{\Omega}}\,|\nabla u_\varepsilon^s|^2\,\mathrm{d}x+\int_{B_r\setminus\overline{\Omega}}\,k^{2}|{u}^{s}_{\varepsilon}|^2\,\mathrm{d}x=\int_{B_r\backslash\overline{\Omega}}\,f\,\,\overline{u}^{s}_{\varepsilon}\,\mathrm{d}x.		\label{equ:2.31}
	\end{aligned}\end{equation} 
	Then, by adding (\ref{equ:2.30}) and (\ref{equ:2.31}), while considering the boundary condition in (\ref{equ:1.5}), we have
	\begin{equation}	
\begin{aligned}
	&\int_{\partial \Omega}\left(\frac{\partial u_{\varepsilon}^{s}}{\partial\nu}+\frac{\partial u^{i}}{\partial\nu}\right)\left(\overline u_{\varepsilon}^{s}+\overline u^{i}\right)\,\mathrm{d}s-\int_{G}\varepsilon^{-1}|\nabla u_{\varepsilon}|^{2}\,\mathrm{d}x+\int_{G}k^{2}(\eta_0+i\varepsilon^{-1}\tau_0)|{u}_{\varepsilon}|^{2}\,\mathrm{d}x
	\\&-\int_{\Omega\setminus\overline{G}}\,\nabla\bar u_{\varepsilon}\cdot A\nabla u_{\varepsilon} \,\mathrm{d}x+\int_{\Omega\setminus\overline{G}}\,k^{2}q(x){u}_{\varepsilon}^{2}\,\mathrm{d}x+\int_{\partial B_r}\,\frac{\partial {u}^{s}_{\varepsilon}}{\partial v}\overline{u}^{s}_{\varepsilon}\,\mathrm{d}s-\int_{\partial \Omega}\,\frac{\partial {u}^{s}_{\varepsilon}}{\partial v}\overline{u}^{s}_{\varepsilon}\,\mathrm{d}s\\&-\int_{B_r\setminus\overline{\Omega}}\,|\nabla u_\varepsilon^s|^2\,\mathrm{d}x+\int_{B_r\setminus\overline{\Omega}}\,k^{2}|{u}^{s}_{\varepsilon}|^2\,\mathrm{d}x=\int_{B_r\backslash\overline{\Omega}}\,f\,\overline{u}^{s}_{\varepsilon}\,\mathrm{d}x.
	\label{equ:2.32}
	\end{aligned}\end{equation}
	By taking the imaginary and real parts of (\ref{equ:2.32}), we can derive the following equations
	\begin{equation*}	
\begin{aligned}&\int_{G}\,k^{2}\varepsilon^{-1}\tau_0|{u}_{\varepsilon}|^{2}\,\mathrm{d}x+\int_{\Omega\setminus\overline{G}}\,k^{2}\Im q(x){u}_{\varepsilon}^{2}\,\mathrm{d}x+\Im\int_{\partial\Omega}\,\frac{\partial u_{\varepsilon}^{s}}{\partial\nu}\bar{u}^{i} \,\mathrm{d}s+\Im\int_{\partial\Omega}\,\frac{\partial u^i}{\partial\nu}\bar{u}_\varepsilon^s\,\mathrm{d}s\\+&\Im\int_{\partial\Omega}\,\frac{\partial u^{i}}{\partial\nu}\bar{u}^{i}\,\mathrm{d}s+\Im\int_{\partial B_{r}}\,\frac{\partial u_{\varepsilon}^{s}}{\partial\nu}\bar{u}_{\varepsilon}^{s}\,\mathrm{d}s-\Im\int_{\Omega\setminus\overline{G}}\,\nabla\bar u_{\varepsilon}\cdot A\nabla u_{\varepsilon} \,\mathrm{d}x=\Im\int_{B_{r}\setminus\overline{\Omega}}\,f\,\bar{u}_{\varepsilon}^{s}\,\mathrm{d}x,
\end{aligned}\end{equation*}
and
\begin{equation*}	
\begin{aligned}
	-&\int_{G}\varepsilon^{-1}|\nabla u_{\varepsilon}|^{2}\,\mathrm{d}x+\int_{G}k^{2}\eta_0|{u}_{\varepsilon}|^{2}\,\mathrm{d}x-\Re\int_{\Omega\setminus\overline{G}}\,\nabla\bar u_{\varepsilon}\cdot A\nabla u_{\varepsilon} \,\mathrm{d}x+\Re\int_{\partial\Omega}\frac{\partial u^{i}}{\partial\nu}\bar{u}^{i}\,\mathrm{d}s\\+&\Re\int_{\partial B_{r}}\frac{\partial u_{\varepsilon}^{s}}{\partial\nu}\bar{u}_{\varepsilon}^{s}\,\mathrm{d}s+\Re\int_{\partial\Omega}\frac{\partial u_{\varepsilon}^{s}}{\partial\nu}\bar{u}^{i} \,\mathrm{d}s+\Re\int_{\partial\Omega}\frac{\partial u^i}{\partial\nu}\bar{u}_\varepsilon^s\,\mathrm{d}s+\int_{\Omega\setminus\overline{G}}\,k^{2}\Re\ q(x)|{u}_{\varepsilon}|^{2}\,\mathrm{d}x\\-&\int_{B_r\setminus\overline{\Omega}}\,|\nabla u_\varepsilon^s|^2\,\mathrm{d}x+\int_{B_r\setminus\overline{\Omega}}\,k^{2}|{u}^{s}_{\varepsilon}|^2\,\mathrm{d}x=\Re\int_{B_{r}\setminus\overline{\Omega}}\,f\,\bar{u}_{\varepsilon}^{s}\,\mathrm{d}x.
	\end{aligned}\end{equation*}
	Then, we can readily obtain	
\begin{align}
\|u_\varepsilon\|_{L^2(G)}^2&\leq C_3\varepsilon\Big(\|u_\varepsilon\|_{H^1(\Omega\setminus\overline{G})}^2+(\|u^{i}\|_{H^1(B_r\setminus\overline{\Omega})}+\|u^{s}_\varepsilon\|_{H^1(B_r\setminus\overline{\Omega})})^2\nonumber\\
&\hspace{4cm}+\|f\|_{L^2(B_r\setminus\overline{\Omega})}\|u^{s}_\varepsilon\|_{H^1(B_r\setminus\overline{\Omega})}\Big)
\nonumber\\&\leq C_4\varepsilon\left(\|u_\varepsilon\|_{H^1(B_r\setminus\overline{G})}^2+\|u^{i}\|_{H^1(B_r\setminus\overline{\Omega})}^2+\|f\|_{L^2(B_r\setminus\overline{\Omega})}^2\right).\label{equ:2.35}\end{align}
\begin{equation}\begin{aligned}
	\|\nabla u_\varepsilon\|_{L^2(G)}^2&\leq C_5\varepsilon\left(\| u_\varepsilon\|_{L^2(G)}^2+\| u_\varepsilon\|_{H^1(B_r\setminus\overline{G})}^2+\|u^{i}\|_{H^1(B_r\setminus\overline{\Omega})}^2+\|f\|_{L^2(B_r\setminus\overline{\Omega})}^2\right),\label{equ:2.36}\hspace*{15mm}
\end{aligned}\end{equation}
where $ C_3$, $ C_4$, and $ C_5$ are positive constants depending on $q$, $\eta_0$,$\tau_0$, $B_r$, $k$, $\Omega$, $\gamma$, and $G$.
 For sufficiently small $\varepsilon$,  it holds that $\varepsilon^2+\varepsilon \leq 2\varepsilon$.	The following result follows by  (\ref{equ:2.35}) and (\ref{equ:2.36}),
\begin{equation}	
\begin{aligned}
	\|u_\varepsilon\|_{H^1(G)}\leq C_6\varepsilon^{\frac{1}{2} }\left(\| u_\varepsilon\|_{H^1(B_r\setminus\overline{G})}^2+\|u^{i}\|_{H^1(B_r\setminus\overline{\Omega})}^2+\|f\|_{L^2(B_r\setminus\overline{\Omega})}^2\right)^{\frac{1}{2} },	\label{equ:2.38}	\end{aligned}\end{equation} where $C_6$ is a positive constant not relying on $\varepsilon$.
	
Next, we begin the proof that the inequality (\ref{equ:2.26}) holds. Suppose that (\ref{equ:2.26}) is not true. Without loss of generality, we assume that for any nonnegative integer $n\in\mathbb{N}$,  there exists a set of data $(f^n,u^{i}_n, {u^n_\varepsilon})$, where ${u^n_\varepsilon}$ is the unique solution of (\ref{equ:1.5}) with 
	$f^n,u^{i}_n$ as inputs and satisfies the restrictions
	\begin{equation*}
\begin{cases}\|f^n\|_{L^2(B_{r_0}\setminus\overline{\Omega})}+\|u_n^{i}\|_{H^1(B_r\setminus\overline{\Omega})}=1,\vspace*{3mm}\\
	\|u_{\varepsilon}^{n}\|_{H^{1}(B_r\setminus\overline{G})}\to\infty\quad\text{as}\quad\varepsilon\to0.
\end{cases}
\end{equation*}
Therefore, we can construct another set of data
$(\widetilde{f}^n,\widetilde{g}^{i},\widetilde{g})$ as follows \begin{equation*}
\left\{\begin{array}{ll}\widetilde{f}^{n}=\frac{f^{n}}{\|{u}^{n}_\varepsilon\|_{H^{1}(B_{r}\setminus\overline{G})}},\quad\widetilde{g}=\frac{{u}^{n}_\varepsilon}{\|{u}^{n}_\varepsilon\|_{H^{1}(B_{r}\setminus\overline{G})}},\vspace*{3mm}\\\widetilde{g}^{i}=\frac{u_{n}^{i}}{\|{u}^{n}_\varepsilon\|_{H^{1}(B_{r}\setminus\overline{G})}},\quad\widetilde{g}^{s}=\frac{{u}^{n,s}_\varepsilon}{\|{u}^{n}_\varepsilon\|_{H^{1}(B_{r}\setminus\overline{G})}},\\[3.8ex]{u}^{n,s}_\varepsilon={u}^{n}_{\varepsilon}-u_{n}^{i},\quad\quad\widetilde{g}^{s}=\widetilde{g}-\widetilde{g}^{i}.\end{array}\right.
\end{equation*}
Clearly, $\widetilde{g}\in H_{loc}^1(\mathbb{R}^N)$  is the unique solution of (\ref{equ:1.5}) with the incident wave $\widetilde{g}^{i}$ and the source $\widetilde{f}^n$. Then we have
\begin{equation}\begin{aligned}
	\left\|\widetilde{g}\right\|_{H^{1}(B_{r}\setminus\overline{G})}=1,\quad\|\widetilde{f}^n\|_{L^2(B_{r_0}\setminus\overline{\Omega})}\to0,\quad\|\widetilde{g}^{i}\|_{H^1(B_r\setminus\overline{\Omega})}\to0\quad\mathrm{as}\quad\varepsilon\to0^+.
	\label{equ:2.41}
\end{aligned}
\end{equation}
Hence, one implies that
\begin{equation*}\begin{aligned}
	\|\widetilde{g}\|_{H^1(G)}\leq& C_7\varepsilon^{\frac{1}{2} }\left(\| \widetilde{g}\|_{H^1(B_r\setminus\overline{G})}^2+\|\widetilde{g}^{i}\|_{H^1(B_r\setminus\overline{\Omega})}^2+\|\widetilde{f}^{n}\|_{L^2(B_r\setminus\overline{\Omega})}^2\right)^{\frac{1}{2} }\leq&C_8\varepsilon^{\frac{1}{2} },
\end{aligned}
\end{equation*}
where $C_7$ and $C_8$ are positive constants are positive constants independent of $\varepsilon$. From Lemma \ref{lem:2.3} and the sixth equation of (\ref{equ:1.5}), we conclude that $(\widetilde{g}\big|_{\Omega\setminus\overline{G}}, \widetilde{g}^s\big|_{\mathbb{R}^N\setminus\overline{\Omega}})$  is the unique solution of (\ref{equ:2.1}) with $p=\widetilde{g}\Big|_{\partial G}$, $h_{1}=\widetilde{g}^{i}\big|_{\partial\Omega}$, and $h_2=\left.\frac{\partial\widetilde{g}^{i}}{\partial\nu}\right|_{\partial\Omega}$. By Lemma \ref{lem:2.3}, we obtain 
\begin{equation*} \begin{aligned}
	\|\widetilde{g}\|_{H^{1}(B_{r}\setminus\overline{G})}& \leq C_{9}\left(\|\widetilde{g}\|_{H^{1/2}(\partial G)}+\|\widetilde{f}^{n}\|_{L^{2}(B_{r_{0}}\setminus\overline{\Omega})}+\|\widetilde{g}^{i}\|_{H^{1}(B_{r}\setminus\overline{\Omega})}\right)\\
	&\leq C_{10}\left(	\|\widetilde{g}\|_{H^{1}(G)}+\|\widetilde{f}^{n}\|_{L^{2}(B_{r_{0}}\setminus\overline{\Omega})}+\|\widetilde{g}^{i}\|_{H^{1}(B_{r}\setminus\overline{\Omega})}\right),	
	\end{aligned}\end{equation*}
	where $C_9$ and $C_{10}$ are positive constants not relying on $\varepsilon$. In addition, we observe that  $\left\|\widetilde{g}\right\|_{H^1(B_r\setminus\overline{G})} $ to $ 0$ as $\varepsilon\to 0$, which contradicts the equality  	$\left\|\widetilde{g}\right\|_{H^{1}(B_{r}\setminus\overline{G})}=1$ given in \eqref{equ:2.41}. Hence, the inequality (\ref{equ:2.26}) holds. Further, we try to prove (\ref{equ:2.27}). From (\ref{equ:2.26}) and (\ref{equ:2.38}), it is easy to obtain that (\ref{equ:2.27}) holds. 
\par	Thus we complete the proof.
\end{proof}

\section{Proofs of Theorem \ref{thm:1.1} and Theorem \ref{thm:1.2} }\label{sec:4}
This section will use the lemmas in Section \ref{sec:3} to provide detailed proofs of Theorem \ref{thm:1.1} and Theorem \ref{thm:1.2}. Before proceeding, we first establish the following estimates regarding the solutions $u$ and $u_\varepsilon$ to (\ref{equ:1.2}) and (\ref{equ:1.5}), respectively, as listed below.
\begin{prop}\label{prop:3.1}
Let $u_{\varepsilon}\in H_{loc}^{1}(\mathbb{R}^{N})$ and $u\in H_{loc}^1(\mathbb{R}^N\backslash\overline{G})$ be the solutions to (\ref{equ:1.5}) and (\ref{equ:1.2}) respectively. Then, for any $r>r_{0}$, there exist small $\varepsilon_{0}>0$ and $C>0$ such that the following estimate holds for $\varepsilon<\varepsilon_{0}$:
\begin{equation*}
\|u_{\varepsilon}-u\|_{H^{1}(B_{r}\setminus\overline{G})}\leq C\varepsilon^{1/2}\left(\|u^{i}\|_{H^{1}(B_{r}\setminus\overline{\Omega})}+\|f\|_{L^{2}(B_{r_{0}}\setminus\overline{\Omega})}\right).
\end{equation*} 
\end{prop}
\begin{proof} 
Let $Y=u_{\varepsilon}-u$ and $Y^s=u_\varepsilon^s-u^s$. From the boundary conditions of (\ref{equ:1.2}) and (\ref{equ:1.5}), we can easily verify that $(Y,\,Y^s)$ is the unique solution to (\ref{equ:2.1}) with the boundary conditions $f=h_1=h_2=0$ and $p=Y|_{\partial G}=u^+_{\varepsilon}|_{\partial G}$. From Lemma \ref{lem:2.3}, Trace theorem, and Lemma \ref{lem:2.4}, we obtain 
\begin{equation*} \begin{aligned} \|Y\|_{H^1(B_r\setminus\overline{G})}=\|u_\varepsilon-u\|_{H^1(B_r\setminus\overline{G})}&\leq C_1	\|u^+_{\varepsilon}\|_{H^{1/2}(\partial G)}\leq C_{2}	\|u_{\varepsilon}\|_{H^{1}(G)}\\&\leq C_{3}\varepsilon^{1/2}\left(\|u^{i}\|_{H^{1}(B_{r}\setminus\overline{\Omega})}+\|f\|_{L^{2}(B_{r_{0}}\setminus\overline{\Omega})}\right), 
\end{aligned}\end{equation*}
where $C_1$, $C_{2}$ and $C_{3}$ are positive constants not relying on $\varepsilon$. Thus, we finish the proof.

\end{proof}

 We are in the position to give the proofs of Theorem \ref{thm:1.1} and Theorem \ref{thm:1.2}.
\begin{proof}[Proof of Theorem~\ref{thm:1.1}]
We utilize the following explicit expressions for the scattering amplitude of $u^s_\varepsilon$
\begin{equation*}\begin{aligned}
	{u}^{\infty}_\varepsilon(\hat{x})=\psi\int_{\partial B_{r}}\Big\{{u}^{s}_\varepsilon(y) \frac{\partial e^{-ik \hat{x}\cdot y}}{\partial\nu(y)}-\frac{\partial{u}^{s}_\varepsilon(y)}{\partial\nu(y)} e^{-i k \hat{x}\cdot y}\Big\}\mathrm{d}s(y)\quad\hat{x}\in\mathbb{S}^{N-1},	
	\end{aligned}\end{equation*}
    \begin{equation*}\begin{aligned}u^{\infty}(\hat{x})=\psi\int_{\partial B_{r}}\Big\{u^{s}(y) \frac{\partial e^{-ik \hat{x}\cdot y}}{\partial\nu(y)}-\frac{\partial u^{s}(y)}{\partial\nu(y)}e^{-ik\hat{x}\cdot y}\Big\}\,\mathrm{d}s(y)\quad\hat{x}\in\mathbb{S}^{N-1},	
	\end{aligned}\end{equation*}
	where $\psi=\frac1{4\pi}$ for N=3, $\psi=\frac{e^{i\pi/4}}{\sqrt{8\pi k}}$ for N=2.
	Then
	\begin{equation*}\begin{aligned}{u}^{\infty}_\varepsilon(\hat{x})-u^{\infty}(\hat{x})=\psi\int_{\partial B_{r}}\Big\{Y^{s}(y) \frac{\partial e^{-i k \hat{x}\cdot y}}{\partial\nu(y)}-\frac{\partial Y^{s}(y)}{\partial\nu(y)} e^{-i k \hat{x}\cdot y}\Big\}\mathrm{d}s(y)\quad\hat{x}\in\mathbb{S}^{N-1}.	
	\end{aligned}\end{equation*}
Note that \begin{center}
$\begin{cases}\left|\frac{\partial e^{-ik\hat{x}\cdot y}}{\partial\nu}\right|\leq M_1r, \\[2.5ex] \left|e^{-ik\hat{x}\cdot y}\right|~\leq M_2,\end{cases}$
\end{center}
where $\left|\cdot\right|$ denotes the Frobenius norm for a matrix or the Euclidean norm for a vector, $M_1$ and $M_2$ are positive constants, and $\nu$ denotes the unit outward normal to the boundary. One can derive
 the following estimate by Trace theorem and Proposition \ref{prop:3.1},
\begin{center}$\begin{aligned}
	\left\|{u}_\varepsilon^{\infty}(\hat{x})-u^{\infty}(\hat{x})\right\|_{C(\mathbb{S}^{N-1})}& \leq M_{3}\left(\|Y\|_{H^{1/2}(B_r\setminus\overline{G})}+\left\|\frac{\partial Y(y)}{\partial\nu(y)}\right\|_{H^{-1/2}(\partial B_{r})}\right) \\
	&\leq M_4\|Y\|_{H^1(B_r\setminus\overline{G})}+M_5\|Y\|_{H^1(B_r\setminus\overline{G})}\\
	&\leq M_{6} \, \varepsilon^{1/2}\left(\|u^{i}\|_{H^{1}(B_{r}\setminus\overline{\Omega})}+\|f\|_{L^{2}(B_{r_{0}}\setminus\overline{\Omega})}\right),
\end{aligned}$\end{center}
where $M_3$, $ M_4 $,  $M_5 $,  $M_6 $ are positive constants depending only on $k$, $\gamma$, $B_{r_{0}}\backslash\overline{\Omega}$, and $B_r\setminus\overline{G}$.
\par Thus, the proof is complete.
\end{proof}
\begin{proof}[Proof of Theorem~\ref{thm:1.2}]
Based on the trace theorem and Lemma \ref{lem:2.4}, we have
\begin{equation*}\begin{aligned}
	\|u^+_\varepsilon\|_{H^{1/2}(\partial G)}\leq\widetilde{C}_{1}\,\|u_\varepsilon\|_{H^1(G)}\leq \widetilde{C}_{2}\,\varepsilon^{\frac{1}{2}}\left(\|f\|_{L^{2}(B_{r_{0}}\setminus\overline{\Omega})}+\|u^{i}\|_{H^{1}(B_{r}\setminus\overline{\Omega})}\right),
    \end{aligned}\end{equation*}
where $\widetilde{C}_{1}$ and  $\widetilde{C}_{2}$ are positive constants not relying on $\varepsilon$.

The proof is complete.
\end{proof}

\section{Sharp estimates of the convergence for a special case}\label{sec:5}
For the system model represented in \eqref{equ:1.2}, we consider a specific scenario in which $G=B_{r_{1}}$ and $\Omega=\mathbb{R}^3$, with $B_{r_1}$ denoting a ball of radius  $r_1$ centered at the origin. In this context, we designate $G$ being filled with an inhomogeneous medium, while the surrounding region is characterized by a homogeneous medium. Our focus is on the scattering phenomena induced by incident plane waves, with the source term  $f = 0$. Our goal is to establish the validity of Theorem \ref{thm:1.1} and Theorem \ref{thm:1.2} within this distinct three-dimensional context, as encapsulated in Theorem \ref{thm:4.1} and Theorem \ref{thm:4.2} below. A similar methodology can also be extended to the two-dimensional scenario. By combining the above assumptions, we can reformulate \eqref{equ:1.5} and \eqref{equ:1.2} in the following forms. Let $u_{\varepsilon}(x) \in H_{loc}^{1}(\mathbb{R}^{3})$ satisfy the following system:
\begin{equation}\begin{cases}\nabla\cdot(\varepsilon^{-1}\nabla u_{\varepsilon})+k^{2}(\eta_{0}+i\varepsilon^{-1}\tau_{0})u_{\varepsilon}=0     ~~~\quad&\text{in}~~G,\vspace*{1mm}\\\Delta u_{\varepsilon}+k^{2}u_{\varepsilon}=0&\text{in}~~\mathbb{R}^{3}\setminus\overline{G},\vspace*{1mm}\\u_{\varepsilon}(x)=e^{ikx\cdot d}+u_{\varepsilon}^{s}(x)&\text{in}~~\mathbb{R}^{3}\setminus\overline{G},\vspace*{1mm}\\u_{\varepsilon}^{-}=u_{\varepsilon}^{+}, \quad \varepsilon^{-1}\frac{\partial u_{\varepsilon}^{-}}{\partial\nu}=\frac{\partial u_{\varepsilon}^{+}}{\partial\nu}&\text{on}~~\partial G, \vspace*{2mm}\\\lim\limits_{|x|\to\infty}|x|\left\{\frac{\partial u_{\varepsilon}^{s}}{\partial|x|}-iku_{\varepsilon}^{s}\right\}=0.\label{equ:4.1}\end{cases}\end{equation}
Additionally, (\ref{equ:1.2}) can be reduced to the following form\begin{equation}\begin{cases}
\Delta u+k^2u=0&\text{in}~ \mathbb{R}^3\setminus\overline{G},\vspace*{1mm}\\u(x)=e^{ikx\cdot d}+u^s(x)&\text{in}~ \mathbb{R}^3\setminus\overline{G},\vspace*{1mm}\\u=0&\text{on} ~\partial G,\vspace*{1mm}\\\lim\limits_{|x|\to\infty}|x|\left\{\frac{\partial u^s}{\partial|x|}-iku^s\right\}=0.
\label{equ:4.2}\end{cases}\end{equation}
Define $q_0 := (\eta_0 + i \varepsilon^{-1}\tau_0)\varepsilon = \varepsilon\eta_0 + i\tau_0 = r e^{i\theta}$, where $r = \sqrt{\varepsilon^2 \eta_0^2 + \tau_0^2}$ and $\theta = \operatorname{Arg}(q_0)$. The complex roots are given by $w_t = \sqrt{r} e^{i \left( \frac{\theta}{2} + t\pi \right)}$ for $t = 0, 1$. Despite some potential misuse of notation, we adopt $\sqrt{q_0}$ to represent the branch corresponding to $t = 0$. We utilize the spherical wave series expansions of the solutions presented in (\ref{equ:4.1}) and (\ref{equ:4.2}). For a detailed discussion of spherical wave functions, refer to \cite{ref2}. Let $u_{\varepsilon}\left(x\right) $, $u_{\varepsilon}^{s}\left(x\right)$ and $u^{s}\left(x\right)$ be given by 
\begin{align}
u_{\varepsilon}(x)&=\sum_{n=0}^{\infty}\sum_{m=-n}^{n}b_{n}^{m} j_{n}(k\sqrt{q_{0}}|x|)Y_{n}^{m}(\hat{x})~\quad x\in B_{r_{1}}, 	\nonumber\\
\quad u_{\varepsilon}^{s}(x)&=\sum_{n=0}^{\infty}\sum_{m=-n}^{n}a_{n}^{m}h_{n}^{(1)}(k|x|)Y_{n}^{m}(\hat{x})~\qquad x\in\mathbb{R}^3\setminus\overline{B_{r_{1}} },\label{equ:4.4}\\
u^{s}(x)&=\sum_{n=0}^{\infty}\sum_{m=-n}^{n}c_{n}^{m}h_{n}^{(1)}(k|x|)Y_{n}^{m}(\hat{x})~\quad\quad x\in\mathbb{R}^3\setminus\overline{B_{r_{1}}},\label{equ:4.5}
\end{align}
where $\hat{x}=x/|x|\in\mathbb{S}^2$. The plane wave can be expressed by using the following series representation:
\begin{equation}\begin{aligned}
e^{ikx\cdot d}=\sum_{n=0}^{\infty}\sum_{m=-n}^{n}i^{n}4\pi\overline{Y_{n}^{m}(d)}j_{n}(k|x|)Y_{n}^{m}(\hat{x}),	\label{equ:4.6}\end{aligned}\end{equation}
where $d$ is the direction of propagation. From (\ref{equ:4.5}), (\ref{equ:4.6}) and the third equation of (\ref{equ:4.2}), we obtain  \begin{equation}\begin{aligned}
c_n^m=\frac{-i^n4\pi\overline{Y_n^m(d)}j_n(kr_1)}{h_n^{(1)}(kr_1)}.
\label{equ:4.7}\end{aligned}\end{equation}
Next, by applying the transmission boundary conditions in  (\ref{equ:4.1}) and comparing the coefficients of  $Y_n^m(\hat{x})$, we derive
\begin{equation}\begin{cases}b_n^mj_n(k\sqrt{q_0}r_1)=a_n^mh_n^{(1)}(kr_1)+i^n4\pi\overline{Y_n^m(d)}j_n(kr_1),\vspace*{5mm}\\
k\sqrt{q_0}b_n^mj_n^{\prime}(k\sqrt{q_0}r_1)=\varepsilon\left(ka_n^mh_n^{(1)'}(kr_1)+4k\pi i^n\overline{Y_n^m(d)}j_n^{\prime}(kr_1)\right).\label{equ:4.8}\end{cases}\end{equation}
By solving (\ref{equ:4.8}), we obtain
\begin{equation}
\begin{cases}
a_n^m=-\frac{4\pi i^n\overline{Y_{n}^{m}(d)}\left(\varepsilon j_n(k\sqrt{q_0}r_1 )j_n^\prime(kr_1)-\sqrt{q_0}j_n(kr_1)j_n^\prime(k\sqrt{q_0}r_1)\right)}{\varepsilon j_n(k\sqrt{q_0}r_1)h_n^{(1)\prime}(kr_1)-\sqrt{q_0}h_n^{(1)}(kr_1)j_n^\prime(k\sqrt{q_0}r_1)} , \vspace{5mm}\\
b_n^m=\frac{4\pi i^n\varepsilon\overline{Y_{n}^{m}(d)}\left( j_n(kr_1)h_n^{(1)\prime}(kr_1)-h_n^{(1)}(kr_1)j_n^{\prime}(kr_1 )\right)}{\varepsilon j_n(k\sqrt{q_0}r_1)h_n^{(1)\prime}(kr_1)-\sqrt{q_0}h_n^{(1)}(kr_1)j_n^\prime(k\sqrt{q_0}r_1)}. 
\label{equ:4.9} \end{cases}\end{equation}
Then, we shall show the sharpness of the estimates in Theorem \ref{thm:1.1}.
\begin{thm}
Under the above of assumptions, the far-field patterns ${u}_\varepsilon^{\infty}{~and~}{u}^{\infty}$ corresponding to the solutions $u_{\varepsilon}~and~u$ of
systems (\ref{equ:4.1}) and (\ref{equ:4.2}), satisfy the following relationship \begin{equation*}\begin{aligned}\begin{vmatrix}{u}_\varepsilon^{\infty}(\hat{x})-u^{\infty}(\hat{x})\end{vmatrix}=C_\mathcal{\infty} \varepsilon^{1/2}+\mathcal{O}(\varepsilon), \quad\forall\hat{x}\in\mathbb{S}^2, 
\end{aligned}\end{equation*}
where $C_\mathcal{\infty}$ depends only on $\eta_0$, $\tau_0$, $k$, $r_1$, and $d$.\label{thm:4.1}\end{thm}\begin{proof} 
According to (\ref{equ:4.4}), (\ref{equ:4.5}), and (\ref{equ:4.7}), we have 
\begin{equation*}\begin{aligned}
	{u}_\varepsilon^{\infty}(\hat{x})=& \frac{1}{k}\sum_{n=0}^{\infty}\sum_{m=-n}^{n}\frac{1}{i^{n+1}}a_{n}^{m}Y_{n}^{m}(\hat{x}), \\
	{u}^{\infty}(\hat{x})=& \frac{i}{k}\sum_{n=0}^{\infty}4\pi \frac{j_{n}(kr_{1})}{h_{n}^{(1)}(kr_{1})}\sum_{m=-n}^{n}\overline{{Y_{n}^{m}(d)}}Y_{n}^{m}(\hat{x}). 
    \end{aligned}\end{equation*}
	Furthermore, we can obtain
\begin{equation*}
\begin{vmatrix}{u}_\varepsilon^{\infty}(\hat{x})-{u}^{\infty}(\hat{x})\end{vmatrix}=\begin{vmatrix}\frac{i}{k}\sum_{n=0}^{\infty}\sum_{m=-n}^{n}(\frac{-1}{i^{n}}a_{n}^{m}-4\pi \frac{j_{n}(kr_{1})}{h_{n}^{(1)}(kr_{1})}\overline{{Y_{n}^{m}(d)}})Y_{n}^{m}(\hat{x})\end{vmatrix}.
\end{equation*}
By a series of calculations, we have 
\begin{align*}
\begin{vmatrix}\frac{-1}{i^{n}}a_{n}^{m}-4\pi \frac{j_{n}(kr_{1})}{h_{n}^{(1)}(kr_{1})}\overline{{Y_{n}^{m}(d)}}\end{vmatrix}=&\begin{vmatrix}\frac{4\pi\left(-\varepsilon j_n(kr_1)j_n(k\sqrt{q_0}r_1)h_{n}^{(1)\prime}(kr_{1})+\varepsilon  h_{n}^{(1)}(kr_{1})j_n(k\sqrt{q_0}r_1)j_n^{\prime}(kr_1)\right)}{h_n^{(1)}(kr_1)\left(-\varepsilon j_n(k\sqrt{q_0}r_1)h_n^{(1)\prime}(kr_1)+\sqrt{q_0} h_n^{(1)}(kr_1)j_n^{\prime}(k\sqrt{q_0}r_1)\right)}\end{vmatrix}\\
&\times \begin{vmatrix}\overline{{Y_{n}^{m}(d)}}\end{vmatrix}.
\end{align*}
By using the Wronskian  $j_{n}(t)y_{n}^{\prime}(t)-j_{n}^{\prime}(t)y_{n}(t)=1/t^{2}$, then we have 	\begin{equation*}\begin{aligned}\begin{vmatrix}\frac{-1}{i^{n}}a_{n}^{m}-4\pi \frac{j_{n}(kr_{1})}{h_{n}^{(1)}(kr_{1})}\overline{{Y_{n}^{m}(d)}}\end{vmatrix}=&\begin{vmatrix}\frac{4\pi i\varepsilon j_n(k\sqrt{q_{0}}r_{1})}{k^{2}r_{1}^{2}h_n^{(1)}(kr_1)\left(-\varepsilon j_n(k\sqrt{q_0}r_1)h_n^{(1)\prime}(kr_1)+\sqrt{q_{0}} h_n^{(1)}(kr_1)j_n^{\prime}(k\sqrt{q_0}r_1)\right)}\end{vmatrix}\\
&\times	\begin{vmatrix}\overline{{Y_{n}^{m}(d)}}\end{vmatrix}.	
\end{aligned}\end{equation*}
From the following relation 
\begin{equation*}\begin{aligned}P_n^{\prime}(z)=-P_{n+1}(z)+\frac{n}{z}P_{n}(z)
\end{aligned}\end{equation*}
where $P_n(z)=j_n(z)$ or $ h_n^{(1)}(z)$, we have
\begin{equation*}\begin{aligned}
	\frac{\mathrm P_n'(k\sqrt{q_0}r_1)}{\mathrm P_n(k\sqrt{q_0}r_1)}\cdot\sqrt{q_{0}}=\frac{n}{kr_{1}}-\sqrt{q_{0}}\frac{P_{n+1}(k\sqrt{q_{0}}r_{1})}{P_n(k\sqrt{q_{0}}r_{1})}.
    \end{aligned}\end{equation*}
Clearly, based on the asymptotic behavior of $j_n(t)$ and $h_{n}^{(1)}(t)$, we have
\begin{align}
	j_n(t)&=\frac{t^n}{1\cdot3\cdots(2n + 1)}\left(1+\mathcal{O}\Big(\frac{1}{n}\Big)\right),\quad n\to\infty\nonumber\\
	h_n^{(1)}(t)&=\frac{1\cdot3\cdots(2n - 1)}{it^{n + 1}}\left(1+\mathcal{O}\left(\frac{1}{n}\right)\right),\quad n\to\infty.\label{equ:4.15}
\end{align}

Therefore, by utilizing \cite{ref15}, we can conclude that
\begin{equation*}\begin{aligned}\begin{vmatrix}\frac{-1}{i^{n}}a_{n}^{m}-4\pi \frac{j_{n}(kr_{1})}{h_{n}^{(1)}(kr_{1})}\overline{{Y_{n}^{m}(d)}}\end{vmatrix}=&\begin{vmatrix}\frac{4\pi i\varepsilon}{k^{2}r_{1}^{2}h_{n}^{(1)}(kr_{1})\left(-\varepsilon h_{n}^{(1)\prime}(kr_{1})+\sqrt{q_0}h_{n}^{(1)}(kr_{1})\frac{\mathrm j_n'(k\sqrt{q_0}r_1)}{\mathrm j_n(k\sqrt{q_0}r_1)}\right)}\end{vmatrix}\begin{vmatrix}\overline{{Y_{n}^{m}(d)}}\end{vmatrix}\\=&\begin{vmatrix}\frac{4\pi i\varepsilon}{k^{2}r_{1}^{2}h_{n}^{(1)}(kr_{1})\left(-\varepsilon h_{n}^{(1)\prime}(kr_{1})+(\frac{n}{kr_1}-\sqrt{q_0}\frac{k\sqrt{q_0}r_1}{2n+3})h_{n}^{(1)}(kr_{1})\right)}\end{vmatrix}\begin{vmatrix}\overline{{Y_{n}^{m}(d)}}\end{vmatrix}\\=&\begin{vmatrix}\frac{4\pi i\varepsilon}{k^{2}r_{1}^{2}h_{n}^{(1)\prime}(kr_{1})\left(-\varepsilon +(\frac{n}{kr_1}-\frac{k{q_0}r_1}{2n+3})\frac{h_{n}^{(1)}(kr_{1})}{h_{n}^{(1)\prime}(kr_{1})}\right)h_n^{(1)}(kr_1)}\end{vmatrix}\begin{vmatrix}\overline{{Y_{n}^{m}(d)}}\end{vmatrix}
	\\\leq&\begin{vmatrix}\frac{4\pi i}{k^{2}r_{1}^{2}h_{n}^{(1)}(kr_{1})h_{n}^{(1)^{\prime}}(kr_{1})}\end{vmatrix}\begin{vmatrix}\frac{\varepsilon}{\varepsilon\left(1-\frac{k^{2}r_{1}^{2}\eta_0}{(2n+3)(n+1)}\right)-i\frac{k^{2}r_{1}^{2}\tau_{0}}{(2n+3)(n+1)}}\end{vmatrix}\begin{vmatrix}\overline{{Y_{n}^{m}(d)}}\end{vmatrix}\\\leq&\begin{vmatrix}\frac{4\pi i}{k^{2}r_{1}^{2}h_{n}^{(1)}(kr_{1})h_{n}^{(1)^{\prime}}(kr_{1})}\end{vmatrix}\begin{vmatrix}\frac{\varepsilon}{a\varepsilon-ib}\end{vmatrix}
	\begin{vmatrix}\overline{{Y_{n}^{m}(d)}}\end{vmatrix}
	\\\leq&\begin{vmatrix}\frac{4\pi i}{k^{2}r_{1}^{2}h_{n}^{(1)}(kr_{1})h_{n}^{(1)^{\prime}}(kr_{1})}\end{vmatrix}\frac{\varepsilon}{\sqrt{2a\varepsilon b}}\begin{vmatrix}\overline{{Y_{n}^{m}(d)}}\end{vmatrix}\\=&\begin{vmatrix}\frac{4\pi i}{k^{2}r_{1}^{2}h_{n}^{(1)}(kr_{1})h_{n}^{(1)^{\prime}}(kr_{1})}\end{vmatrix}\frac{1}{\sqrt{2ab}}\varepsilon^{\frac{1}{2}}\begin{vmatrix}\overline{{Y_{n}^{m}(d)}}\end{vmatrix},
    \end{aligned}\end{equation*}
	where $a=1-\frac{k^{2}r_{1}^{2}\eta_0}{(2n+3)(n+1)}$ and $b=\frac{k^{2}r_{1}^{2}\tau_{0}}{(2n+3)(n+1)}.$ Furthermore, it is easy to derive
	\begin{equation*}\begin{aligned}\begin{vmatrix}\left(\frac{-1}{i^{n}}a_{n}^{m}-4\pi \frac{j_{n}(kr_{1})}{h_{n}^{(1)}(kr_{1})}\overline{{Y_{n}^{m}(d)}}\right)Y_n^m(\hat{x})\end{vmatrix}\sim& \begin{vmatrix}\frac{4\pi }{k^{2}r_{1}^{2}h_{n}^{(1)}(kr_{1})h_{n}^{(1)^{\prime}}(kr_{1})}\frac{1}{\sqrt{2ab}}\varepsilon^{\frac{1}{2}}\end{vmatrix}\begin{vmatrix}\overline{{Y_{n}^{m}(d)}}Y_n^m(\hat{x})\end{vmatrix}\\&\hspace{6cm}+O(\varepsilon). 
    \end{aligned}\end{equation*}
    Note that for any $n,m\in\mathbb{N}$,
    \begin{equation*}\begin{aligned}|\overline{Y_n^m(d)}Y_n^m(\hat{x})|\leq\frac{2n+1}{4\pi}.
    \end{aligned}\end{equation*}
	Therefore,
	\begin{center}$\begin{vmatrix}{u}_\varepsilon^{\infty}(\hat{x})-{u}^{\infty}(\hat{x})\end{vmatrix}\sim\begin{vmatrix}\frac{i}{k}\sum\limits_{n=0}^\infty\sum\limits_{m=-n}^n\frac{4\pi }{k^{2}r_{1}^{2}h_{n}^{(1)}(kr_{1})h_{n}^{(1)^{\prime}}(kr_{1})}\frac{1}{\sqrt{2ab}}\varepsilon^{\frac{1}{2}}\end{vmatrix}\begin{vmatrix}\frac{2n+1}{4\pi}\end{vmatrix},$\end{center}where $C_{\mathcal{\infty}}$ is chosen as
    \begin{center}$C_{\mathcal{\infty}}=\begin{vmatrix}\sum\limits_{n=0}^\infty\sum\limits_{m=-n}^n\frac{2n+1}{k^{3}r_{1}^{2}h_{n}^{(1)}(kr_{1})h_{n}^{(1)^{\prime}}(kr_{1})}\frac{1}{\sqrt{2ab}}\end{vmatrix}.$\end{center}Then, we finish the proof.\end{proof} 
The subsequent theorem formulates that the wave field  $u_{\varepsilon}$  exhibits an attenuation characteristic within the region of $\partial B_{r_{1}}$. 
\begin{thm}
There exists a constant $C_v$ which depends only on $k$, $r_1$, $d$, $\eta_0$, and $\tau_0$ such that 
\begin{equation*}\begin{aligned}	\left\| u_\varepsilon^+\right\|_{H^{1/2}(\partial B_{r_1})}=C_v \varepsilon^{1/2}+\mathcal{O}(\varepsilon).
\end{aligned}\end{equation*}\label{thm:4.2}
\end{thm}
\begin{proof}\vspace{-2.6mm}From (\ref{equ:4.9}) and the Wronskian, we know that
\begin{equation*}\begin{aligned} b_n^m =&\frac{4\pi i^{n}\varepsilon\overline{Y_{n}^{m}(d)}\Big(j_{n}(kr_{1})\left(j_{n}^{\prime}(kr_{1})+iy_{n}^{\prime}(kr_{1})\right)-\left(j_{n}(kr_{1})+iy_{n}(kr_{1})\right)j_{n}^{\prime}(kr_{1})\Big)}{-\sqrt{q_{0}}j_{n}^{\prime}(k\sqrt{q_{0}}r_{1})h_{n}^{(1)}(kr_{1})+\varepsilon j_{n}(k\sqrt{q_{0}}r_{1})h_{n}^{(1)^{\prime}}(kr_{1})}\\
	=&\frac{4\pi\varepsilon i^{n+1}\overline{Y_{n}^{m}(d)}\left(j_{n}(kr_{1})y_{n}^{\prime}(kr_{1})-j_{n}^{\prime}(kr_{1})y_{n}(kr_{1})\right)}{-\sqrt{q_{0}}j_{n}^{\prime}(k\sqrt{q_{0}}r_{1})h_{n}^{(1)}(kr_{1})+\varepsilon j_{n}(k\sqrt{q_{0}}r_{1})h_{n}^{(1)^{\prime}}(kr_{1})}\\
	=&\frac{4\pi i^{n+1}\varepsilon\overline{Y_{n}^{m}(d)}}{k^2r^2_{1}\left(-\sqrt{q_{0}}j_{n}^{\prime}(k\sqrt{q_{0}}r_{1})h_{n}^{(1)}(kr_{1})+\varepsilon j_{n}(k\sqrt{q_{0}}r_{1})h_{n}^{(1)^{\prime}}(kr_{1})\right)}\\
	=&\frac{4\pi i^{n+1}\overline{Y_{n}^{m}(d)}}{k^2r^2_{1}(-\sqrt{q_{0}}\frac{\mathrm j_n'(k\sqrt{q_0}r_1)}{\mathrm j_n(k\sqrt{q_0}r_1)}\frac{h_{n}^{(1)}(kr_{1})}{h_{n}^{(1)^{\prime}}(kr_{1})}+\varepsilon)}\frac{\varepsilon}{j_{n}(k\sqrt{q_{0}}r_{1})h_{n}^{(1)^{\prime}}(kr_{1})}.
	\end{aligned}\end{equation*}
	\vspace{-1.8mm} Based on the asymptotic behavior of  $h_{n}^{(1)}(kr_{1})$ in (\ref{equ:4.15}), we obtain
	\begin{equation*}\begin{aligned}\left|b_n^{m}j_n(k\sqrt{q_{0}}r_{1})\right|=&\left| \frac{4\pi i^{n}\overline{Y_{n}^{m}(d)}}{k^2r^2_{1}h_{n}^{(1)^{\prime}}(kr_{1})}\frac{i\varepsilon}{-\sqrt{q_{0}}\left(\frac{n}{k\sqrt{q_{0}r_{1}}}-\frac{kr_{1}\sqrt{q_{0}}}{2n+3}\right)(\frac{-kr_{1}}{n+1})+\varepsilon }       \right|
	\\=&\left| \frac{4\pi i^{n}\overline{Y_{n}^{m}(d)}}{{k^2r^2_{1}h_{n}^{(1)^{\prime}}(kr_{1})}}\frac{i\varepsilon}{\varepsilon\left(1-\frac{k^{2}r_{1}^{2}\eta_0}{(2n+3)(n+1)}\right)-\frac{i\tau_0k^2r^2_{1}}{(2n+3)(n+1)}}\right|
	\\=&  \left| \frac{4\pi i^{n}\overline{Y_{n}^{m}(d)}}{{k^2r^2_{1}h_{n}^{(1)^{\prime}}(kr_{1})}}\right| \left|\frac{i\varepsilon}{a\varepsilon -ib}\right|
	\\\leq&\left|\frac{4\pi i^{n}\overline{Y_{n}^{m}(d)}}{{k^2r^2_{1}h_{n}^{(1)^{\prime}}(kr_{1})}}\right| \frac{\varepsilon}{\sqrt{2\varepsilon ab}}\\\leq&\left|\frac{4\pi i^{n}\overline{Y_{n}^{m}(d)}}{{k^2r^2_{1}h_{n}^{(1)^{\prime}}(kr_{1})}}\right| \frac{1}{\sqrt{2ab}}\varepsilon^\frac{1}{2}, 
    \end{aligned}\end{equation*}
	where $a=1-\frac{k^{2}r_{1}^{2}\eta_0}{(2n+3)(n+1)}$ and $b=\frac{\tau_0k^2r^2_{1}}{(2n+3)(n+1)}$. Namely, \begin{equation*}\begin{aligned}
	\left|b_n^{m}j_n(k\sqrt{q_{0}}r_{1})\right|=\left|\frac{4\pi i^{n}\overline{Y_{n}^{m}(d)}}{{k^2r^2_{1}h_{n}^{(1)^{\prime}}(kr_{1})}}\right| \frac{1}{\sqrt{2ab}}\varepsilon^{1/2}+\mathcal{O}(\varepsilon). 
    \end{aligned}\end{equation*}
	According to \cite{ref4}, through direct calculations we obtain\begin{equation*}\begin{aligned}
	\left\|u_{\varepsilon}^{+}\right\|_{H^{1/2}(\partial B_{r_{1}})}=\Bigg(\sum_{n=0}^\infty\sum_{m=-n}^n\left(1+n^2\right)^{1/2}|b_n^mj_n(k\sqrt{q_0}r_1)|^2\Bigg)^{1/2}.
    \end{aligned}\end{equation*}
    Besides, let $h_v$=$\sum\limits_{n=0}^\infty\sum\limits_{m=-n}^n\left(1+n^2\right)^{1/2}\left|\frac{4\pi i^{n}\overline{Y_{n}^{m}(d)}}{{k^2r_{1}^2h_{n}^{(1)^{\prime}}(kr_{1})}}\right|^2 \frac{1}{2ab}$ in the same time, we can choose $C_{v}$=$\sqrt{h_v}$.
	\par Thus we complete the proof of Theorem \ref{thm:4.2}.\end{proof}
\section{Appendix}
This section focuses on analyzing the well-posedness of the scattering problem \eqref{equ:1.5}, which is essential for proving Theorem \ref{thm:1.1}. To establish the uniqueness of the solution to \eqref{equ:1.5}, we refer to {\it Step I} in Lemma \ref{lem:2.3}, which provides a similar proof and shall not be repeated here. Next, we turn our attention to demonstrating both the existence of the solution and its continuous dependence on the input data. Our approach is based on Fredholm theory.
\begin{proof}Set	$$\gamma_{\varepsilon(x)}=\begin{cases}\varepsilon^{-1}&x\in{G},\\[1.5mm]A(x)&x\in\Omega\backslash\overline{G}\end{cases}\quad\mbox{and}\quad q_{\varepsilon(x)}=\begin{cases}\eta_{0}+i\varepsilon^{-1}\tau_{0}&x\in{G},\\[1.5ex]q(x)&x\in\Omega\backslash\overline{G}.\end{cases}$$ 
Then we can obtain \begin{equation}
\begin{cases}\nabla\cdot(\gamma_{\varepsilon(x)}\nabla u_{\varepsilon})+k^{2} q_{\varepsilon(x)}u_{\varepsilon}=0&\text{in}~ \Omega,\vspace*{1mm}\\
	\Delta u_{\varepsilon}^{s}+k^{2}u_{\varepsilon}^{s}=f&\text{in} ~\mathbb{R}^{N}\setminus\overline{\Omega},\vspace*{1mm}\\
	u_{\varepsilon}^{-}=u_{\varepsilon}^{+},\quad\varepsilon^{-1} \frac{\partial u_{\varepsilon}^{-}}{\partial\nu}=\frac{\partial u_{\varepsilon}^{+}}{\partial\nu_{A}}&\text{on}~ \partial G,\vspace*{2mm}\\
	u_{\varepsilon}^{-}=u_{\varepsilon}^{s}+g_1,\quad \frac{\partial u_{\varepsilon}^{-}}{\partial\nu_{A}}=\frac{\partial u_{\varepsilon}^{s}}{\partial\nu}+g_2&\text{on}~ \partial\Omega,\vspace*{2mm}\\
	\lim\limits_{ |x|\to\infty} ~|x|^{(N-1)/2}\left\{\frac{\partial u_{\varepsilon}^{s}}{\partial|x|}-iku_{\varepsilon}^{s}\right\}=0,\label{equ:5.1}
\end{cases}
\end{equation} where  $g_1\in H^{\frac{1}{2}}(\partial\Omega)$ and $g_2\in H^{-\frac{1}{2}}(\partial\Omega)$.  
By taking a proper truncation of $\mathbb{R}^3\backslash\overline{G}$, we obtain
\begin{equation}
\begin{cases}\nabla\cdot(\gamma_{\varepsilon(x)}\nabla u_{\varepsilon})+k^{2} q_{\varepsilon(x)}u_{\varepsilon}=0&\text{in}~ \Omega,\vspace*{1mm} \\
	\Delta u_{\varepsilon}^{s}+k^{2}u_{\varepsilon}^{s}=f&\text{in} ~B_r\setminus\overline{\Omega}, \vspace*{1mm}\\
	u_{\varepsilon}^{-}=u_{\varepsilon}^{+},\quad\varepsilon^{-1} \frac{\partial u_{\varepsilon}^{-}}{\partial\nu}= \frac{\partial u_{\varepsilon}^{+}}{\partial\nu_{A}}&\text{on}~ \partial G, \vspace*{2mm}\\
	u_{\varepsilon}^{-}=u_{\varepsilon}^{s}+g_1,\quad \frac{\partial u_{\varepsilon}^{-}}{\partial\nu_{A}}=\frac{\partial u_{\varepsilon}^{s}}{\partial\nu}+g_2&\text{on}~ \partial\Omega, \vspace*{2mm}\\
	\frac{\partial u_{\varepsilon}^{s}}{\partial\nu}=\Lambda u_{\varepsilon}^s&\text{on}~ \partial B_r. \label{equ:5.2}
\end{cases}
\end{equation} 
Clearly,  (\ref{equ:5.1}) and (\ref{equ:5.2}) are equivalent. Let $w=u_{\varepsilon}$ in $\Omega$ and $w=u_{\varepsilon}^s+\tilde{v}$ in $B_r\backslash\overline\Omega$. Here, we assume that $k^2$ is not a Dirichlet eigenvalue of $-\bigtriangleup$ in $B_r\backslash\overline\Omega$ and that $\tilde{v}\in H^1(B_r\backslash\overline{\Omega})$ satisfies 	
\begin{equation*}
\begin{cases}\Delta\widetilde{v}+k^2\widetilde{v}=0&\text{in~}B_r\backslash\overline{\Omega},\\[2pt]
	\widetilde{v}=g_1&\text{on }\partial\Omega,\\[2pt]
	\widetilde{v}=0&\text{on}~\partial B_r.
\end{cases}
\end{equation*}
Furthermore, we can express (\ref{equ:5.2}) as
\begin{equation*}
\begin{cases}
	\nabla\cdot(\gamma_{\varepsilon(x)}\nabla w)+k^{2}q_{\varepsilon(x)}w=0&\text{in}~ \Omega,\vspace*{2mm}\\
	\Delta w+k^{2}w=f&\text{in} ~B_r\setminus\overline{\Omega},\vspace*{2mm}\\
	w^{-}=w^{+},\quad\varepsilon^{-1} \frac{\partial w^{-}}{\partial\nu}=\frac{\partial w^{+}}{\partial\nu_{A}}&\text{on}~\partial G, \vspace*{2mm}\\
	w^{-}=w^{+},\quad\frac{\partial w^{-}}{\partial\nu_{A}}=\frac{\partial w^{+}}{\partial\nu}-\frac{\partial \tilde{v}}{\partial\nu}+g_2&\text{on}~ \partial\Omega,\vspace*{2mm}\\
	\frac{\partial w}{\partial\nu}=\Lambda w+\frac{\partial \tilde{v}}{\partial\nu}&\text{on}~ \partial B_r.
\end{cases}
\end{equation*}
As discussed regarding the $DtN$ operator for $\Lambda$ and $\Lambda_{0}$ in Section \ref{sec:3},  it follows immediately that for any $\varphi\in H^{1}(B_r)$, there exists $w\in H^{1}(B_r)$ such that
\begin{equation*} 
a_1(w,\varphi)+a_2(w,\varphi)=\mathcal{F}(\varphi),
\end{equation*}
where 
\begin{align*}
a_1(w,\varphi):=&\int_G{\varepsilon^{-1}\nabla w\cdot\nabla\bar{\varphi}}\,\mathrm{d}x+\int_{\Omega\backslash\overline G}\,\nabla\bar{\varphi}\cdot A \nabla w\,\mathrm{d}x+\int_{\Omega\setminus\overline{G}}\, k^2w\bar{\varphi}\,\mathrm{d}x\\&+\int_{B_r\setminus\overline{\Omega}}\,\nabla w\cdot\nabla\bar{\varphi}\,\mathrm{d}x+\int_{B_r\setminus\overline{\Omega}}\,k^2w\bar{\varphi}\,\mathrm{d}x-\int_{\partial B_r}\Lambda_0w\,\bar{\varphi}\,\mathrm{d}s,\\ a_2(w,\varphi):=&-\int_{\Omega\setminus\overline{G}}k^2(q+1)w\bar{\varphi}\,\mathrm{d}x-\int_{G}k^2(\eta_{0}+i\varepsilon^{-1}\tau_{0})w\bar{\varphi}\,\mathrm{d}x\\&-2\int_{B_{r}\setminus\overline{\Omega}} k^{2}w\bar{\varphi}\,\mathrm{d}x-\int_{\partial B_{r}} (\Lambda-\Lambda_{0})w\bar{\varphi}\,\mathrm{d}s,\\
\mathcal{F}(\varphi):=&\int_{\partial\Omega}(g_{2}-\frac{\partial\tilde{v}}{\partial\nu})\bar{\varphi}\,\mathrm{d}s+\int_{\partial B_{r}}\frac{\partial\tilde{v}}{\partial\nu}\bar{\varphi}\,\mathrm{d}s-\int_{B_r\setminus\overline{\Omega}}f\bar{\varphi}\,\mathrm{d}x.
\end{align*}
Similarly, by employing the same proof method as in Lemma \ref{lem:2.3}, we can identify the bounded operator $T$  and the compact operators $T_1$ and $T_2$, such that \begin{equation*}
\begin{aligned}
	a_1(w,\varphi)=\left \langle Tw,\varphi \right \rangle, 
    \end{aligned}\end{equation*}
	with		
\begin{align*} 
a_3(w,\varphi):&=\int_{\Omega\setminus\overline{G}}k^2(q+1)w\bar{\varphi}\,\mathrm{d}x+\int_{G}k^2(\eta_{0}+i\varepsilon^{-1}\tau_{0})w\bar{\varphi}\,\mathrm{d}x+2\int_{B_{r}\setminus\overline{\Omega}} k^{2}w\bar{\varphi}\,\mathrm{d}x=\left \langle T_1w,\varphi \right \rangle,\\
\vspace{-3mm}a_4(w,\varphi):&=\int_{\partial B_{r}} (\Lambda-\Lambda_{0})w\bar{\varphi}\,\mathrm{d}s=\left \langle T_2w,\varphi \right \rangle.
\end{align*}
In this context, $a_1(\cdot,\cdot)$, $a_2(\cdot,\cdot)$, $a_3(\cdot,\cdot)$, and $a_4(\cdot,\cdot)$ are continuous on $H^1(B_r)\times H^1(B_r)$, the notation $\left \langle\cdot,\cdot  \right \rangle $ denotes the inner product in 
$H^1(B_r)$. Since ${T}$ is bounded and invertible, ${T}_{1}+{T}_{2}$ is compact, therefore, ${T}-({T}_1+{T}_2)$ is a Fredholm operator of index zero. Furthermore, $\left({T}-({T}_1+{T}_2)\right)^{-1}$  is bounded. It follows that
\begin{equation*}\begin{aligned}\left|\mathcal{F}(\varphi)\right|\leq C_1'\left(\|g_2\|_{H^{-1/2}(\partial\Omega)}+\|g_1\|_{H^{1/2}(\partial\Omega)}+\|f\|_{L^{2}(B_{r_0}\setminus\overline\Omega)}\right)\|\varphi\|_{H^1({B_r\backslash\overline{G}})}.
\end{aligned}\end{equation*} Directly, we obtain
\begin{equation*}\begin{aligned} \left\|u_\varepsilon\right\|_{H^1(\Omega)}\leq C_2'\left(\|g_2\|_{H^{-1/2}(\partial\Omega)}+\|g_1\|_{H^{1/2}(\partial\Omega)}+\|f\|_{L^{2}(B_{r_0}\setminus\overline\Omega)}\right), 
\end{aligned}\end{equation*}
where $C_1'$ and $C_2'$ are positive constants independent on $\varepsilon$. Finally, we finish the proof of the well-posedness of the PDE system (\ref{equ:1.5}).
\end{proof}

\noindent\textbf{Acknowledgment.} 
The work of H. Diao is supported by the National Natural Science Foundation of China  (No. 12371422) and the Fundamental Research Funds for the Central Universities, JLU.  The work of Q. Meng is supported by a fellowship award from the Research Grants Council of the Hong Kong Special Administrative Region, China (Project No. CityU PDFS2324-1S09). Also the work of Q. Meng is supported by the Hong Kong RGC General Research Funds (projects 11311122, 11300821, and 11303125),  the NSFC/RGC Joint Research Fund (project N\_CityU101/21), and the France-Hong Kong ANR/RGC Joint Research Grant, A-CityU203/19..


\begin{thebibliography}{99}	
\bibitem{Cakoni2003}F. Cakoni and D. Colton, \textit{A uniqueness theorem for an inverse electromagnetic scattering problem in inhomogeneous anisotropic media}, Proceedings of the Edinburgh Mathematical Society, {\bf 46(2)}(2003), 293--314.
\bibitem{ref4} F. Cakoni and D. Colton, \textit{Qualitative methods in inverse scattering theory: An introduction}, Springer Science \& Business Media, 2005.
\bibitem{Cakoni2004} F. Cakoni, D. Colton, and P. Monk, \textit{The electromagnetic inverse-scattering problem for partly coated Lipschitz domains}, Proceedings of the Royal Society of Edinburgh Section A: Mathematics, {\bf 134(4)}(2004), 661--682.



%

\bibitem{ref2} D. Colton and R. Kress, \textit{Inverse acoustic and electromagnetic scattering theory}, Berlin: Springer, 1998.

\bibitem{DFLW2022} H. Diao, X. Fei, H. Liu, and L. Wang, {\it Determining anomalies in a semilinear elliptic equation by a minimal number of measurements}, Inverse Problems, {\bf 41}(2025), 055004.
    
\bibitem{DGT2025}H. Diao, Y. Geng, and R. Tang, {\it Non-radiating elastic sources in inhomogeneous elastic media at corners with applications}, Inverse Problems, {\bf 41}(2025), 085013.


\bibitem{Hahenr93} 
P. H\"{a}hner, \textit{A uniqueness theorem for a transmission problem in inverse electromagnetic scattering}, Inverse Problems, \textbf{9(6)}(1993), 667.
\bibitem{Haher1998}
P. H\"{a}hner, \textit{On acoustic, electromagnetic, and elastic scattering problems in inhomogeneous media}, Habilitation thesis, Universit\"{a}t G\"{o}ttingen, Mathematisches Institute, 1998.


\bibitem{ref6}P. H\"{a}hner, \textit{On the uniqueness of the shape of a penetrable, anisotropic obstacle}, Journal of computational and applied mathematics, {\bf 116(1)}(2000), 167--180.

\bibitem{ref99} V. Isakov, \textit{On uniqueness in the inverse transmission scattering problem}, Communications in Partial Differential Equations, {\bf 15(11)}(1990), 1565--1586.

\bibitem{Kk1993} A.	Kirsch and R. Kress, {\it Uniqueness in inverse obstacle scattering}, Inverse Problems, \textbf{9}(1993), 285--299.

\bibitem{KOVW2020} R. V. Kohn, D. Onofrei, M. S. Vogelius,  and M. I. Weinstein, \textit{Cloaking via change of variables for the Helmholtz equation}, Communications on Pure and Applied Mathematics, {\bf 63}(2020), 973--1016.

\bibitem{KW2021}P.-Z. Kow and J.-N. Wang, {\it Reconstruction of an impenetrable obstacle in anisotropic inhomogeneous background}, IMA Journal of Applied Mathematics, {\bf 86}(2021), 320--348. 
\bibitem{LAX1967} P. D. Lax and R. S. Phillips, {\textit Scattering Theory}, Academic Press, New York, 1967.
 
 \bibitem{Liu2009}H. Liu,  {\it Virtual reshaping and invisibility in obstacle scattering},  Inverse Problems, {\bf 25} (2009), 045006.

\bibitem{LL17}
H. Liu and X. Liu, \textit{Recovery of an embedded obstacle and its surrounding medium from formally determined scattering data}, Inverse Problems, {\bf 33(6)}(2017), 065001.

\bibitem{Liu2012} 
H. Liu, Z. Shang, H. Sun, and J. Zou, \textit{Singular perturbation of reduced wave equation and scattering from an embedded obstacle}, Journal of Dynamics and Differential Equations, {\bf 24(4)}(2012), 803--821.

\bibitem{LiuSun2013} H. Liu and H. Sun, {\it Enhanced near-cloak by FSH lining}, Journal de Math\'{e}matiques Pures et Appliqu\'{e}es, {\bf 99}(2013), 17--42.
\bibitem{LZZ2015} H. Liu, H. Zhao, and C. Zou, \textit{Determining scattering support of anisotropic acoustic mediums and obstacles}, Communications in Mathematical Sciences, {\bf 13(4)}(2015), 987--1000.

\bibitem{LZ2006} H.	Liu and J. Zou, {\it Uniqueness in an inverse acoustic obstacle scattering problem for both 
sound-hard and sound-soft polyhedral scatterers}, Inverse Problems, \textbf{22}(2006), 515--524.


\bibitem{Nguyen2010} H.-M. Nguyen,\textit{ Cloaking via change of variables for the Helmholtz equation in the whole space}, Communications on Pure and Applied Mathematics,  \textbf{ 63}(2010), 1505--1524.

\bibitem{ref14}
R. L. Ochs, Jr., \textit{The limited aperture problem of inverse acoustic scattering: Dirichlet boundary conditions}, SIAM Journal on Applied Mathematics, {\bf 47(6)}(1987), 1320-1341.
\bibitem{ref15}
F. Olver, D. W. Lozier, R. F. Boisvert, and C. W. Clark, \textit{NIST Handbook of Mathematical Functions Hardback and CD-ROM}, Cambridge University Press, 2010. 

\bibitem{SU1987} J. Sylvester and G. Uhlmann, {\it A global uniqueness theorem for an inverse boundary value problem}, Ann. Math., \textbf{125}(1987), 153--169.

\bibitem{ref13}
V. H. Weston, \textit{Multifrequency inverse problem for the reduced wave equation: resolution cell and stability}, Journal of mathematical physics, {\bf 25(12)}(1984), 3483--3488.

\bibitem{YZ}J. Yang, B. Zhang, and H. Zhang, \textit{Uniqueness in inverse acoustic and electromagnetic scattering by penetrable obstacles with embedded objects}, Journal of Differential Equations, {\bf 265(12)}(2018), 6352--6383. 

\bibitem{Z}	H. Zhang and B. Zhang, \textit{A Newton method for a simultaneous reconstruction of an interface and a buried obstacle from far-field data}, Inverse Problems, {\bf 29(4)}(2013), 045009.
\end{thebibliography}
\end{document}